\documentclass[11pt]{article}
\usepackage[utf8]{inputenc}

\usepackage{graphicx}
\usepackage{xcolor}
\usepackage{amsmath,amsthm,amssymb}
\usepackage[margin=1in]{geometry}
\usepackage[pdftex,colorlinks=true,linkcolor=blue,citecolor=blue,urlcolor=black]{hyperref}
\usepackage{comment}
\usepackage[normalem]{ulem}

\def\b{{\bf b}}
\def\x{{\bf x}}

\newcommand{\ket}[1]{| #1 \rangle}

\DeclareMathOperator{\poly}{poly}

\newcommand{\be}{\begin{equation}}
\newcommand{\ee}{\end{equation}}
\newcommand{\bea}{\begin{eqnarray}}
\newcommand{\eea}{\end{eqnarray}}
\newcommand{\bes}{\begin{equation*}}
\newcommand{\ees}{\end{equation*}}
\newcommand{\beas}{\begin{eqnarray*}}
\newcommand{\eeas}{\end{eqnarray*}}

\makeatletter
\newtheorem*{rep@theorem}{\rep@title}
\newcommand{\newreptheorem}[2]{%
\newenvironment{rep#1}[1]{%
 \def\rep@title{#2 \ref{##1} (restated)}%
 \begin{rep@theorem}}%
 {\end{rep@theorem}}}
\makeatother

\allowdisplaybreaks

\newtheorem{thm}{Theorem}
\newtheorem*{thm*}{Theorem}
\newtheorem{cor}[thm]{Corollary}

\newtheorem{lem}[thm]{Lemma}
\newtheorem*{lem*}{Lemma}
\newtheorem{claim}[thm]{Claim}
\newtheorem{prop}[thm]{Proposition}

\newtheorem{rem}[thm]{Remark}

\newreptheorem{thm}{Theorem}
\newreptheorem{lem}{Lemma}
\newreptheorem{cor}{Corollary}
\newtheorem{prob}{Problem}

\def\ds{\displaystyle}

\usepackage{cite}
\usepackage{algorithm}
\usepackage{algorithmic}

\makeatletter
\newenvironment{breakablealgorithm}
  {
   \begin{center}
     \refstepcounter{algorithm}
     \hrule height1pt depth0pt \kern3pt
     \renewcommand{\caption}[2][\relax]{
       {\raggedright\textbf{\ALG@name~\thealgorithm} ##2\par}%
       \ifx\relax##1\relax 
         \addcontentsline{loa}{algorithm}{\protect\numberline{\thealgorithm}##2}%
       \else 
         \addcontentsline{loa}{algorithm}{\protect\numberline{\thealgorithm}##1}%
       \fi
       \kern3pt\hrule\kern3pt
     }
  }{
     \kern3pt\hrule\relax%
   \end{center}
  }
\makeatother

\title{Computing eigenvalues of diagonalizable matrices in a quantum computer}
\author{Changpeng Shao\footnote{changpeng.shao@bristol.ac.uk} \\
School of Mathematics, University of Bristol, Bristol, UK}

\date{\today}

\begin{document}

\maketitle

\begin{abstract}

Solving linear systems and computing eigenvalues are two fundamental problems in linear algebra.
For solving linear systems, many efficient quantum algorithms have been discovered. 
For computing eigenvalues, currently, we have efficient quantum algorithms for Hermitian and unitary matrices. However, the general case is far from fully understood.
Combining quantum phase estimation, quantum algorithm to solve linear differential equations and quantum singular value estimation,
we propose two quantum algorithms to compute the eigenvalues of diagonalizable matrices that only have real eigenvalues and 
normal matrices.
The output of the quantum algorithms is a superposition of the eigenvalues and the corresponding eigenvectors.
The complexities are dominated by 
solving a linear system of ODEs and 
performing quantum singular value estimation, which usually can be solved efficiently in a quantum computer.
In the special case when the matrix $M$ is $s$-sparse, the complexity is $\widetilde{O}(s\rho^2 \kappa^2/\epsilon^2)$
for diagonalizable matrices that only have real eigenvalues, and 
$\widetilde{O}(s\rho\|M\|_{\max} /\epsilon^2)$ for
normal matrices. Here $\rho$ is an upper
bound of the eigenvalues, $\kappa$ is the conditioning of the eigenvalue problem, and $\epsilon$ is the precision to approximate the eigenvalues. We also extend the quantum algorithm to diagonalizable matrices with complex eigenvalues under an extra assumption.
\end{abstract}

\section{Introduction}

Solving linear systems and computing eigenvalues are two fundamental problems in linear algebra. 
These are two problems of major importance in many
scientific and engineering applications \cite{arbenz2012lecture,saad2011numerical,van1983matrix}.
Many classical algorithms were discovered in the past decades.
Except for some special cases, the theoretical complexities of these algorithms to solve linear systems of size $n$ or to
estimate the eigenvalues of $n\times n$ matrices range from $O(n^2)$ to $O(n^3)$.
On the other hand,
recent developments on quantum algorithms \cite{HHL,block-encoding,childs2017quantum,AbramsLloyd,berry2017quantum,berry2014high,childs2019quantum,kerenidis_et_al:LIPIcs:2017:8154} 
verify the fact that quantum computers can solve many linear algebraic problems much faster than classical computers. 
For instance, under certain conditions we can efficiently estimate the eigenvalues of unitary matrices \cite{Kitaev} and Hermitian matrices \cite{AbramsLloyd}, the singular values of general matrices \cite{block-encoding,kerenidis_et_al:LIPIcs:2017:8154}.
For solving linear systems in the quantum version, people have already found many efficient quantum algorithms (e.g. see \cite{HHL,childs2017quantum,block-encoding}).

Currently, the solving of linear systems in a quantum computer is almost optimal, especially by the recently discovered block-encoding method \cite{block-encoding}.
For computing eigenvalues, there are some attempts \cite{Daskin,wang} to estimate the eigenvalues of non-Hermitian and non-unitary matrices. However, this problem is still far from well understood. In this paper, we  focus on the eigenvalue problem of diagonalizable matrices.

The classical eigenvalue algorithms are mainly iterative.
Based on quantum linear algebraic techniques, we may generalize
them into quantum algorithms. However, in this paper, we choose to
follow the idea of quantum phase estimation (QPE) \cite{Kitaev}, a method
that is different from any classical algorithms to estimate eigenvalues.
Following the idea of QPE, the problem we want to
solve in this paper is stated as follows:

\begin{prob}
\label{main problem}
Let $M$ be an $n$-by-$n$ diagonalizable matrix,  $\epsilon\in(0,1)$. Assume that the eigen-pairs of $M$ are $\{(\lambda_j,|E_j\rangle): j=1,\ldots,n\}$.
Given access to copies of the state of the form $\sum_{j=1}^n \beta_j |E_j\rangle$, find a quantum algorithm that outputs a state proportional to
$\sum_{j=1}^n \beta_j |\tilde{\lambda}_j\rangle |E_j\rangle$ such that $|\lambda_j - \tilde{\lambda}_j| \leq \epsilon$.
\end{prob}

The diagonalizability assumption implies that any vector is a linear combination of the eigenvectors, so we can choose any desired vector as the initial state. Usually, the decomposition is not known for us, so the quantum algorithm we construct should be independent of this decomposition.

The eigenvalue problem we solve here is a little different from the classical sense. 
In Problem \ref{main problem}, we only obtain a superposition of the eigen-pairs. We may view it as a quantum version of eigenvalue decomposition. 
There are two reasons for us to solve this problem. First, similar to solve linear systems in a quantum computer, it is better to obtain a quantum result for the eigenvalue problem. 
For us, this should be a quantum state that contains the information of the eigenvalues and the corresponding eigenvectors.
Second, by measuring the first register, we can obtain all the eigenvalues. This needs at least $\Omega(n)$ measurements. 
We can also obtain the quantum states of the corresponding eigenvectors. 
However, we may hope the quantum eigenvalue decomposition has other applications. For instance, implementing functions of matrices. 
Let $f$ be a function, $M=E\Lambda E^{-1}$ be the eigenvalue decomposition of $M$. 
In many cases (e.g. see \cite{higham2008functions}), $f(M) = E f(\Lambda) E^{-1}$. In a quantum computer, by implementing a function of $M$, we mean to implement $f(M)|\b\rangle$ for any given state $|\b\rangle$. 
Since $E$ is non-singular, we have a decomposition $|\b\rangle = \sum_{j=1}^n \beta_j |E_j\rangle$. 
Thus $f(M)|\b\rangle = \sum_{j=1}^n \beta_j f(\lambda_j) |E_j\rangle$. If we can solve Problem \ref{main problem}, then we can construct the state $f(M)|\b\rangle$.

\subsection{Our results}

In this paper, we mainly focus on two types of 
matrices: diagonalizable matrices that only have real eigenvalues and normal matrices. 
We  also consider the case that the diagonalizable matrices can have complex eigenvalues. In this case, we need to make an extra assumption in Problem \ref{main problem}.

Our quantum algorithm is inspired by QPE and recent works in quantum algorithms to solve
ordinary differential equations \cite{berry2014high,
berry2017quantum,childs2019quantum}.
In QPE, an important step is to perform
Hamiltonian simulation to implement $e^{i H t}$ for some Hermitian matrix $H$.
For example, suppose the input state is
$|E_j\rangle$, then after Hamiltonian simulation, we obtain a state
of the form $(1/\sqrt{m})\sum_{l=0}^{m-1} e^{i\lambda_j l} |l\rangle |E_j\rangle$.
Then $\lambda_j$ is recovered by applying quantum
Fourier inverse transform to $|l\rangle$.
When $M$ is not Hermitian, $e^{i M t}$
is not unitary, so we cannot apply Hamiltonian simulation.
One alternative method we can use is
ordinary differential equations.
Just like Hamiltonian simulation, 
which arises from Schr\"{o}dinger equation, 
we can consider
the ODE $d\x(t)/dt = i M \x(t)$ as a source to implement $e^{i M t}$.
It is easy to show that the solution is $\x(t) = e^{iMt} |E_j\rangle$.
In a quantum computer, to solve this ODE,
we usually use the discretization method to 
change it into a linear system. Then solve the linear system by quantum linear algebraic techniques.
An interesting point of this idea is that
the solution of the linear system
is  a superposition of the 
solutions at differential times $t_l = l \Delta t$,
i.e., $\sum_{l=0}^{m-1} |l\rangle |\x(t_l)\rangle = \sum_{l=0}^{m-1} e^{i\lambda_j t_l} |l\rangle |E_j\rangle$.
This state is exactly what we need in QPE.
Consequently, we can obtain the state $|\lambda_j\rangle|E_j\rangle$ by applying quantum Fourier inverse transform to $|l\rangle$. This idea also works when the input state is a linear combinations of the eigenvectors $|\phi\rangle = \sum_j\beta_j|E_j\rangle$. Moreover, in this algorithm, we do not need to know how $|\phi\rangle$ is decomposed.

With the above idea, for the first type of matrices, our main result is

\begin{thm}[Informal of Theorem \ref{thm2}]
Let $M$ be an $n\times n$ diagonalizable matrix which only has real eigenvalues, $\rho\geq 1$ be a upper
bounded of the eigenvalues.
Then
Problem \ref{main problem} can be solved
in time
\be 
O\left( T\left(
\frac{\rho \kappa(E)\log(\rho /\epsilon)}{\epsilon},
\frac{\epsilon}{\rho},
\frac{n\rho\log(\rho /\epsilon)}{\epsilon}\right)
\times \frac{\rho\kappa(E)\sqrt{\log(\rho /\epsilon)}}{\epsilon}  \right),
\ee
where  $T(\cdot,\cdot,\cdot)$ is the complexity to solve certain linear systems, and
$\kappa(E)$ is the condition number of the matrix generated by 
the eigenvectors.
\end{thm}

In a quantum computer, to solve an $n\times n$ linear system, the complexity $T(\kappa,\epsilon,n)$ depends on the condition number $\kappa$ of the coefficient matrix, the precision $\epsilon$ to approximate the solution, the size $n$ of the linear system and some other parameters about the linear system. In many case, $T$ is linear at $\kappa$ and logarithmic at $\epsilon$ and $n$. For instance, if the linear system is $s$-sparse, then $T = O(s\kappa\poly(\log\kappa, \log s, \log 1/\epsilon, \log n)) = \widetilde{O}(s\kappa)$  \cite{childs2017quantum}. Consequently, we have the following result.

\begin{thm}[Informal of Theorem \ref{cor of thm2}]
Assume that $M$ is  diagonalizable and $s$ sparse, $\rho\geq 1$ is a upper
bounded of the eigenvalues. Then
Problem \ref{main problem} can be solved
in time
$
\widetilde{O}(s\rho^2\kappa(E)^2/\epsilon^2).
$
\end{thm}

As discussed above, when we have the state $\sum_{j=1}^n \beta_j |\tilde{\lambda}_j\rangle |E_j\rangle$, we can obtain all the eigenvalues by measurements. It is not hard to show that $O(n\log n)$ measurements are enough. Thus for sparse matrices, the complexity to obtain all the eigenvalues is $
\widetilde{O}(sn\rho^2\kappa(E)^2/\epsilon^2).
$
However, the current classical algorithms cost at least $O(n^2)$ to obtain all the eigenvalues, even if all the eigenvalues are real \cite{huss1997parallelizable,beavers1973computational}. So quantum computers still have the possibility to beat classical computers when the eigenvalue problem is well-conditioned.

The algorithm above also works when $M$ only has purely imaginary eigenvalues.
However, when $M$ has complex eigenvalues, it may not work in a straightforward way. One simple reason is that we cannot estimate the real and imaginary parts of the eigenvalues at the same time in QPE. An alternative approach is to consider the ODE $d\x(t)/dt = i (M\otimes I + I\otimes \overline{M}) \x(t)$ with initial state $\sum_j \beta_j \ket{E_j} \ket{\overline{E}_j}$. Here $\overline{M}, \ket{\overline{E}_j}$ refer to the complex conjugate of $M, \ket{E_j}$ respectively. The solution of this ODE is $\x(t) = \sum_j \beta_j e^{2it {\rm Re}(\lambda_j)} \ket{E_j} \ket{\overline{E}_j}$. Then a similar idea to above shows that we can create the state proportional to $\sum_j \beta_j \ket{{\rm Re}(\lambda_j)}\ket{E_j} \ket{\overline{E}_j}$. If this state is further viewed as the initial state of the ODE $d\x(t)/dt = (M\otimes I - I\otimes \overline{M}) \x(t)$, then we will end up with a state proportional to $\sum_j \beta_j \ket{{\rm Re}(\lambda_j)} \ket{{\rm Im}(\lambda_j)} \ket{E_j} \ket{\overline{E}_j}$. Conclude the above analysis,
we obtain the following theorem.

\begin{thm}[Informal of Theorem \ref{main theorem}]
Let $M$ be an $n\times n$ diagonalizable matrix, $\rho\geq 1$ be a upper bound of the eigenvalues.
Assume that the initial state of Problem \ref{main problem} is $\sum_{j=1}^n\beta_j|E_j\rangle|\overline{E}_j\rangle$, then Problem \ref{main problem}
can be solved in time
\be 
O\left(
T\left(
\frac{\rho \kappa(E)^2\log(\rho /\epsilon)}{\epsilon},
\frac{\epsilon}{\rho},
\frac{n\rho\log(\rho /\epsilon)}{\epsilon}\right)
\times \frac{\rho^2\kappa(E)^4\log(\rho/\epsilon) }{\epsilon^2}   \right),
\ee
Especially, if $M$ has sparsity $s$, then the complexity is 
$
\widetilde{O}(s\rho^3\kappa(E)^6/\epsilon^3).
$
\end{thm}

As for normal matrices, the quantum algorithm is much simpler than above. Since normal matrices are unitarily diagonalizable, the eigenvalue decomposition and singular value decomposition are closely related.
For this kind of matrices, we can directly use quantum singular value estimation (QSVE) to solve Problem \ref{main problem}.
Suppose the eigenvalues are $\sigma_je^{i\theta_j}, j\in\{1,2,\ldots,n\}$.
Since we know how to do QSVE in a quantum
computer \cite{block-encoding,kerenidis_et_al:LIPIcs:2017:8154}, we can estimate $\sigma_j$ through QSVE.
As for $\theta_j$, by considering the QSVE of
$\left(\begin{array}{cc} 
0 & M  \\
M^\dag & 0
\end{array}\right)$,
we can construct a unitary $U$ to
perform $|E_j\rangle \mapsto e^{-i\theta_j} |E_j\rangle$. Hence, $\theta_j$ can be estimated by applying QPE to $U$.

For normal matrices, our main result is stated as follows.

\begin{thm}[Informal of Theorems \ref{thm1}, \ref{cor1}]
\label{intro:thm2}
Let $M$ be an $n\times n$ normal matrix.
Then Problem \ref{main problem} can be solved
in time $O(T/\epsilon^2)$,
where $O(T/\epsilon)$ is the complexity to do quantum singular value estimation of $M$ to precision $\epsilon$.
Especially, if $M$ is $s$ sparse, then the complexity is  $\widetilde{O}(s\rho \|M\|_{\max} /\epsilon^2).$
\end{thm}

In the above theorem, if $M=(m_{ij})$, then $ \|M\|_{\max} = \max_{i,j}|m_{ij}|$. Quantum singular value estimation can be solved efficiently in many situations in a quantum computer, so we can think $T=O({\rm polylog}(n))$.

\subsection{Related works}

The most related work is quantum phase estimation to estimate eigenvalues of Hermitian and unitary matrices \cite{AbramsLloyd,Kitaev}. It is the starting point of our research.
For estimating the eigenvalues of non-unitary (non-Hermitian) matrices, as far as we know, there are two papers \cite{Daskin,wang}. One difficulty behind this might be the result in complexity theory \cite{QMA1, QMA2}, 
which suggesting that 
many eigenvalue problems are QMA-complete, i.e., they are hard to solve even for quantum computers.
In \cite{Daskin}, Daskin et al. 
use the idea of iterative phase estimation algorithm to estimate the complex eigenvalues
of non-unitary matrices. This idea
works when the input state is $|E_j\rangle$.
The complexity of their quantum algorithm is $O(n^2/\epsilon)$.
In \cite{wang}, Wang et al. propose a measurement-based quantum phase estimation algorithm. 
This kind of idea is similar to the power
method. So it usually recovers the largest eigenvalue.
In \cite{low2019hamiltonian}, Low and Chuang studied the qubitization of normal operators. This result might be helpful to improve Theorem \ref{intro:thm2}.

There are many results in the classical case. For example, in a recent paper  \cite{banks2019pseudospectral}, Banks et al. proposed a random classical algorithm that can find an invertible
matrix $V$ and a diagonal matrix $D$ such that
$\|M-VDV^{-1}\|\leq \epsilon$ in time $O(n^\omega \log^2(n/\epsilon))$, where $2\leq \omega < 2.373$ is the exponent 
of matrix multiplication. 
This algorithm is optimal up to polylogarithmic factors, in the sense that verifying that a given similarity diagonalizes a matrix requires at least matrix multiplication time.
As a result, $VDV^{-1}$ can be viewed as an approximated eigenvalue decomposition of $M$.

In physics, eigenvalues of Hamiltonian describe the energy levels
of the quantum system. It is one of the most important tasks in chemistry as they are required to predict reaction rates and electronic structures. The $\mathcal{PT}$-symmetric matrices
refer to a type of Hamiltonians that are not Hermitian but
have real eigenvalues \cite{bender1998real}. These new kinds of Hamiltonians describe a new class of complex quantum theories having positive probabilities and unitary time evolution,
and play crucial roles in quantum mechanics and
many other areas of physics \cite{bender2007making,mostafazadeh2010pseudo}. Our algorithm may provide a method to perform this kind of Hamiltonian simulations in a quantum computer.

\subsection{Organizations of the paper}

In section \ref{sec:Diagonalizable matrices that only have real eigenvalues}, we consider Problem \ref{main problem} for diagonalizable matrices that only have real eigenvalues. Then in section \ref{sec:A quantum algorithm to estimate complex eigenvalues}, we generalize the quantum algorithm to the case when the matrix has complex eigenvalues under an extra assumption. Finally, in section \ref{Normal Matrices}, we studies Problem \ref{main problem} for normal matrices. 
Since QPE is an important idea for our quantum algorithms,
in Appendix \ref{Quantum phase estimation}, we briefly introduce this method and review some quantum linear algebraic techniques that will be used in this paper.

\section{Diagonalizable matrices that only have real eigenvalues}
\label{sec:Diagonalizable matrices that only have real eigenvalues}

In this section, we study
Problem \ref{main problem} in the case that all  the eigenvalues  are real. 
These kind of matrices are closely related to Hermitian matrices. A simple fact is that a diagonalizable matrix $M$ only has real eigenvalues if and only if there is a Hermitian matrix $H$ and an invertible matrix $P$ such that $M = P H P^{-1}$ \cite{drazin1962criteria}. So if $H_1, H_2$ are Hermitian and $H_1$ or $H_2$ is positive definite, then $M = H_1H_2$ only have real eigenvalues. Indeed, this is the only possibility such that a diagonalizable matrix only has real eigenvalues. Note that if we know the decomposition $M = H_1 H_2$, then we can solve Problem \ref{main problem} by QPE directly. More precisely, assume that $H_1$ is positive definite, then there is a positive definite Hermitian matrix $H_3$ such that $H_1 = H_3^2$. Thus $H_3^{-1} M H_3 = H_3 H_2 H_3$ is Hermitian. In a quantum computer, $H_3$ can be created efficiently, such as by the technique of block-encoding \cite{block-encoding}. Now it suffices to find the eigenvalues of the Hermitian matrix $H_3 H_2 H_3$. This can be solved by QPE. In this paper, we will not consider this case.

In the following, we first state the main idea of our quantum algorithm based on the Euler method. Then we show the detailed analysis of solving ODEs based on \cite{berry2017quantum}. As an illustration, we will apply our quantum algorithm to compute the eigenvalues of sparse matrices. Finally, we analyze the difficulties of the quantum algorithm in estimating  complex eigenvalues.


\subsection{Main idea}

Assume that $M$ is an $n\times n$ diagonalizable matrix with real eigenvalues $\{\lambda_1,\ldots,\lambda_n\}$ and eigenvectors $\{|E_1\rangle,\ldots,|E_n\rangle\}$ of unit norm.
For simplicity, we first assume that the input state is $|E_j\rangle$ in Problem \ref{main problem}.
Now let us start from the following linear system of differential equations
\begin{equation} \label{eq1:new}
\begin{cases} \vspace{.2cm}
\displaystyle \frac{d\x(t)}{dt} = 2\pi i M \x(t), \\
\x(0) = |E_j\rangle.
\end{cases}
\end{equation}
It is easy to see that the solution of this differential equation is 
\be \label{eq solu}
\x(t) = e^{2\pi i M t} |E_j\rangle = 
e^{ 2\pi i \lambda_j t}  |E_j\rangle. 
\ee
Since $M$ is not Hermitian, it may not easy to
apply $e^{2\pi i  M t}$ directly to $|E_j\rangle$ 
to obtain the solution. 
Fortunately, there are some
quantum algorithms \cite{berry2014high,berry2017quantum,childs2019quantum}
to solve ordinary differential equations. 
These algorithms are mainly based on the 
discretization method.

One of the simplest discretization method is Euler method:
\be \label{Euler method}
\frac{d\x(t)}{dt} \approx \frac{\x(t) - \x(t-\Delta t)}{ \Delta t}.
\ee
In the following, we briefly show how Euler method is used to solve (\ref{eq1:new}).
First, we discrete the time interval $[0,t]$ into $m$ 
short intervals by setting
$t_0=0, t_1=\Delta t, t_2=2\Delta t, \ldots, t_m=m\Delta t = t$.
After discretization by Euler method (\ref{Euler method}), we obtain a linear system
\be \label{linear system: Euler method}
(A \otimes I_n - 2 \pi i \Delta t I_m\otimes M) \tilde{\x} = \b,
\ee
where
\[
A = \begin{pmatrix}
 1 &  & \\
-1 & 1 & \\
   & \ddots & \ddots & \\
   &        & -1 & 1 
\end{pmatrix}_{m\times m}, \quad
\tilde{\x} = \begin{pmatrix}
\x(t_1) \\
\x(t_2) \\
\vdots \\
\x(t_m)
\end{pmatrix}, \quad
\b = \begin{pmatrix}
\x(t_0) \\
0 \\
\vdots \\
0
\end{pmatrix}.
\]
Solve this linear system in a quantum computer, then we obtain
the quantum state of $\tilde{\x}$
\be
 \frac{1}{\sqrt{m}} \sum_{l=1}^m |l\rangle |\x(t_l)\rangle
=\frac{1}{\sqrt{m}} \sum_{l=1}^m e^{2\pi i l \lambda_j \Delta t} |l\rangle |E_j\rangle.
\ee
This is a quantum state similar to the one required
in the quantum phase estimation (see (\ref{QPE:required state}) in Appendix \ref{Quantum phase estimation}). 
Thus we can obtain $\ket{\tilde{\lambda}_j}\ket{E_j}$
by applying quantum Fourier inverse transform to $|l\rangle$.

If the initial state of (\ref{eq1:new}) 
is a linear combination of the eigenvectors
$\x(0) = \sum_{j=1}^n \beta_j |E_j\rangle$, then
the quantum state of the linear system (\ref{linear system: Euler method}) is proportional to
\be
\sum_{j=1}^n \sum_{l=0}^m  \beta_j e^{2\pi il \lambda_j \Delta t} |l\rangle  |E_j\rangle.
\ee
Similarly, when quantum Fourier inverse transform is applied to $|l\rangle$, we  obtain 
\be \label{general case}
\sum_{j=1}^n  \beta_j  |\tilde{\lambda}_j\rangle  |E_j\rangle.
\ee

In the next section, we will give more details. However,
we will not use the Euler method as it has worse dependence on the
precision \cite{berry2014high}.
Instead, we shall use the quantum algorithm proposed in \cite{berry2017quantum} to solve the ODE (\ref{eq1:new}).
It is poly-log at the precision.

\subsection{Error analysis}

Following the idea discussed above,
we need to solve the differential equation 
(\ref{eq1:new}).
It is interesting to see that the ODE (\ref{eq1:new})
satisfies the assumptions\footnote{ The assumptions are: $2\pi i M$ is diagonalizable, the eigenvalues have non-positive real parts.}
of the quantum algorithm proposed in
\cite{berry2017quantum}.
Thus we can apply this algorithm directly.

In the discretization method, we choose
$t=\Theta(1/\rho\epsilon)$, where $\epsilon$ is the precision
to approximate the eigenvalues, and
$\rho \geq \max\{1, \max_j |\lambda_j|\}$
is an upper bound of the eigenvalues.
The integer $m=\Theta(1/\epsilon)$ and 
$\Delta t = \Theta(1/\rho)$.
The choices of the parameters here 
are inspired by the quantum phase estimation, 
which will become clear later.

For any integer $k$ and any $z \in \mathbb{C}$, denote
the $(k+1)$-terms truncation of exponential function
$e^z$ as
\be
T_k(z) = \sum_{j=0}^k \frac{z^j}{j!}.
\ee
If $r\in \mathbb{R}$ and $|r|<1$, then it is easy to show that
$
|T_k(ir) - e^{ir}| \leq {e}/{(k+1)!}.
$

Set $d=m(k+1)$, define the $(d+1)n\times (d+1)n$ matrix $C_{m,k}$ by
\bea
C_{m,k}(2\pi iM\Delta t) &=& \sum_{p=0}^d |p\rangle \langle p|\otimes I - \sum_{p=0}^{m-1} \sum_{q=1}^k |p(k+1)+q\rangle\langle p(k+1)+q-1|\otimes \frac{2\pi iM\Delta t}{q} \\
&& -\, \sum_{p=0}^{m-1} \sum_{q=0}^k |(p+1)(k+1)\rangle\langle p(k+1)+q|\otimes I.
\label{coeff matrix}
\eea
The goal of this matrix is to implement $T_k(2\pi iM\Delta t)$ without computing matrix powers. For an illustration,
$C_{2,2}(A)$ is of the following form:
\bes
C_{2,2}(A) = 
\left(
\begin{array}{rrrrrrrrrr}
I &  \\
-A  & I \\
  & -\frac{A}{2}  & I \\
-I  & -I  & -I  & I \\
  &   &   &  -A & I \\
  &   &   &   & -\frac{A}{2}  & I \\
  &   &   & -I  & -I  & -I    & \,\,\, I \\
\end{array}
\right).
\ees

The definition of the matrix $C_{m,k}$ is a little different from the one defined in \cite{berry2017quantum}. We deleted their fourth summation term. This is caused by different goals.
In \cite{berry2017quantum}, the authors aim to
compute the quantum state of the solution at time $t_m$. 
The fourth term is introduced to improve the 
probability of this state.
However, in this paper, we need the superposition of the quantum state of the solution at time $t_0,t_1,\ldots,t_m$. Thus, the solution at time $t_m$ is not special for us.

Now we consider the following linear system
\be \label{linear system}
C_{m,k}(2\pi iM\Delta t) \tilde{\x} = |0..0\rangle|E_j\rangle. 
\ee
In the right hand side, $|0..0\rangle$ is added to make sure
both sides have the same dimension. 
Here $\tilde{\x}$ refers to the approximated solution of the differential equation (\ref{eq1:new}).
We can formally write the solution as
\be \label{eq:approsolu}
\tilde{\x} = \sum_{p=0}^{m-1} \sum_{q=0}^k |p(k+1)+q\rangle |\x_{p,q}\rangle + |m(k+1)\rangle |\x_{m,0}\rangle.
\ee
In the above expression, we use the ket notation to simplify the expression of $\tilde{\x}$, but $|\x_{p,q}\rangle$ may not a unit vector.

We can check that (see Appendix \ref{verification})
\bea 
|\x_{0,0}\rangle &=& |E_j\rangle, \label{verify1}\\
|\x_{0,q}\rangle &=& \frac{(2\pi iM\Delta t)^q}{q!}|E_j\rangle =
\frac{(2\pi i\lambda_j \Delta t)^q}{q!}|E_j\rangle, \quad (q=1,\ldots,k). \label{verify2}
\eea
If $p>0$, then
\bea \label{eq:appro solu}
|\x_{p,0}\rangle &=& T_k(2\pi iM\Delta t)|x_{p-1,0}\rangle =
T_k(2\pi i\lambda_j \Delta t)|x_{p-1,0}\rangle, \\
|\x_{p,q}\rangle &=& \frac{(2\pi iM\Delta t)^q}{q!}|x_{p,0}\rangle =
\frac{(2\pi i\lambda_j \Delta t)^q}{q!}|x_{p,0}\rangle, \quad (q=1,\ldots,k). \label{verify3}
\eea

\begin{lem} \label{lemma1}
Assume that $2\pi \Delta t |\lambda_j|<1$, 
then for any 
$p\in\{1,2,\ldots,m\}$, we have
\be
\||\x(t_p)\rangle-|\x_{p,0}\rangle\|_2 \leq \left(1+\frac{e}{(k+1)!}\right)^p -1.
\ee
Moreover,
\be
\Big|\||\x_{p,0}\rangle\|_2 - 1\Big|\leq \left(1+\frac{e}{(k+1)!}\right)^p -1.
\ee
\end{lem}

\begin{proof}
By equations (\ref{eq solu}), (\ref{eq:appro solu}),
\begin{eqnarray*}
\||\x(t_p)\rangle-|\x_{p,0}\rangle\|_2
&=& \|e^{2\pi i\lambda_j \Delta t} |\x(t_{p-1})\rangle - T_k(2\pi i\lambda_j \Delta t)|\x_{p-1,0}\rangle \|_2  \\
&\leq& |e^{2\pi i\lambda_j \Delta t} - T_k(2\pi i\lambda_j \Delta t)|+
|T_k(2\pi i\lambda_j \Delta t)| \, \||\x(t_{p-1})\rangle - |\x_{p-1,0}\rangle\|_2 \\
&\leq& \frac{e}{(k+1)!}+
\left(1+\frac{e}{(k+1)!}\right) \||\x(t_{p-1})\rangle - |\x_{p-1,0}\rangle\|_2 \\
&\leq& \frac{e}{(k+1)!}
\sum_{r=0}^{p-1} \left(1+\frac{e}{(k+1)!}\right)^r \\
&=&
\frac{e}{(k+1)!} \frac{(1+\frac{e}{(k+1)!})^p -1 }{(1+\frac{e}{(k+1)!}) -1} \\
&=& \left(1+\frac{e}{(k+1)!}\right)^p -1.
\end{eqnarray*}
This proves the first claim.
Since $|\x(t_p)\rangle$ is unit, the second result comes naturally.
\end{proof}

\begin{prop} \label{prop:Solution error}
Let $|\x\rangle = \frac{1}{\sqrt{m+1}} \sum_{p=0}^m |p\rangle |\x(t_p)\rangle$ be the superposition of the exact solution of the differential equation (\ref{eq1:new}). Denote 
the quantum state of $\sum_{p=0}^{m} |p\rangle |\x_{p,0}\rangle$ as 
$|\hat{\x}\rangle$,
then
\be
\||\x\rangle - |\hat{\x}\rangle \|_2 \leq
2\sqrt{m+1} \left(\left(1+\frac{e}{(k+1)!}\right)^m - 1\right) \left(2 - \left(1+\frac{e}{(k+1)!}\right)^m\right)^{-1/2}.
\ee
\end{prop}

\begin{proof}
Denote the $l_2$-norm of $\sum_{p=0}^{m} |p\rangle |\x_{p,0}\rangle$ as $\sqrt{L}$. Then
\begin{eqnarray}
\||\x\rangle - |\hat{\x}\rangle \|_2 &=&
\sum_{p=0}^m \left\| \frac{|\x(t_p)\rangle}{\sqrt{m+1}} - \frac{|\x_{p,0}\rangle}{\sqrt{L}} \right\|_2 \nonumber \\
&\leq& \sum_{p=0}^m \left| \frac{1}{\sqrt{m+1}} - \frac{1}{\sqrt{L}} \right|
+\frac{1}{\sqrt{L}} \sum_{p=0}^m 
\| |\x(t_p)\rangle - |\x_{p,0}\rangle \|_2.
\label{proof:prop:eq1}
\end{eqnarray}
By Lemma \ref{lemma1},
\[
|L-m-1| \leq (m+1) \left(\left(1+\frac{e}{(k+1)!}\right)^m - 1\right)=:(m+1)A,
\]
so
\[
\left| \frac{1}{\sqrt{m+1}} - \frac{1}{\sqrt{L}} \right|
=\frac{|L-m-1|}{\sqrt{L(m+1)}(\sqrt{m+1}+\sqrt{L})}
\leq \frac{A}{\sqrt{(m+1)(1-A)}}.
\]
Thus the first term of (\ref{proof:prop:eq1}) is bounded by 
$A \sqrt{m+1}/\sqrt{1-A}$.
Again by Lemma \ref{lemma1}, the second term
of (\ref{proof:prop:eq1}) is also 
bounded by $A\sqrt{m+1}/\sqrt{1-A}$.
This completes the proof.
\end{proof}

Proposition \ref{prop:Solution error} states that if we
choose $k$ such that $(k+1)!\gg m$, then 
the quantum state of the solution of the differential equation
(\ref{eq1:new}) is approximated by the part of the quantum state of
the solution of the linear system (\ref{linear system}) with $q=0$.
The following lemma further enhances this by showing that this part occupies a constant
amplitude. If we perform amplitude amplification, 
then we can enlarge the amplitude of this state close to 1 with $O(1)$ repetitions.

\begin{lem} \label{lemma2}
Assume that $|2\pi \lambda_j \Delta t|<1$, then
the $l_2$-norm of the solution $\tilde{\x}$ of the linear system (\ref{linear system}) satisfies
\be
\left|\|\tilde{\x}\|_2^2 - e^{2 \pi \lambda_j \Delta t} \sum_{p=0}^{m-1}  \||\x_{p,0}\rangle\|_2^2 - \|\x_{m,0}\|_2^2\right|
\leq \frac{e}{(k+1)!}.
\ee
\end{lem}

\begin{proof}
By equation (\ref{eq:approsolu})
\beas 
\|\tilde{\x}\|_2^2 &=& \sum_{p=0}^{m-1} \sum_{q=0}^k  \||\x_{p,q}\rangle\|_2^2 + \|\x_{m,0}\|_2^2 \\  
&=&  \sum_{p=0}^{m-1} \sum_{q=0}^k \frac{(2 \pi\lambda_j \Delta t)^q}{q!} \||\x_{p,0}\rangle\|_2^2  + \|\x_{m,0}\|_2^2 \\
&=&   T_k(2 \pi\lambda_j \Delta t) \sum_{p=0}^{m-1}  \||\x_{p,0}\rangle\|_2^2+ \|\x_{m,0}\|_2^2.
\eeas
It is easy to prove that
$
|T_k(2 \pi\lambda_j \Delta t) - e^{2 \pi\lambda_j \Delta t}| \leq {e}/{(k+1)!}.
$
This completes the proof.
\end{proof}

\subsection{Main results}

In this paper, for solving an $n$-by-$n$ linear system $A\x = \b$ in a quantum computer, we shall use $T(\kappa(A), \epsilon, n)$ to denote the complexity. Here $\kappa(A)$ is the condition number of the $A$ and $\epsilon$ is the precision to approximate the solution state. The cost $T$ may depend on other parameters, like the sparsity or the norm of $A$, but we concern more about $\kappa(A),\epsilon$ and $n$ in this paper. In many cases, for a quantum linear solver, the complexity $T$ is linear at $\kappa(A)$ and poly-log at $\epsilon,n$.

The following theorem concludes the result of solving the ODE (\ref{eq1:new}) 
in the special case when $\x(0) = |E_j\rangle$ for some $j$.

\begin{prop} \label{main prop}
Let $M$ be an $n\times n$ diagonalizable matrix which only has real eigenvalues, $\rho\geq 1$ be a upper bound of the eigenvalues. 
Suppose that the linear system (\ref{coeff matrix}) can be solved in time $T(\kappa(C_{m,k}),\epsilon,m(k+1)n)$ to precision $\epsilon$. Given access to copies of one of the eigenvectors, then there is a quantum algorithm that estimates the corresponding eigenvalue up to precision $\rho \epsilon$ in time
$
T(\kappa(E)\epsilon^{-1}(\log 1/\epsilon),\epsilon,n\epsilon^{-1}\log(1/\epsilon))$,
where $\kappa(E)$ is the condition number of the matrix generated by the eigenvectors.
\end{prop}

\begin{proof}
By Proposition \ref{prop:Solution error}, 
to make sure the error between
$|\x\rangle$ and $|\hat{\x}\rangle$
is smaller than $\epsilon$, we can choose
$k$ such that $(k+1)!\geq m^2/\epsilon$. 
That is $k \geq \log (m/\epsilon)$. Thus,
we can set $k=O(\log (m/\epsilon))=O(\log (1/\epsilon))$
as $m=O(1/\epsilon)$.

Let $|\bar{\x}\rangle$ be the quantum state obtained
by the quantum linear solver to solve the linear system
(\ref{coeff matrix}).
Since $2\pi \Delta  t|\lambda_j|<1$, Lemma \ref{lemma2}
shows that in $|\bar{\x}\rangle$ the amplitude of $|\hat{\x}\rangle$ is close to $e^{-2\pi \Delta  t|\lambda_j|}>e^{-1}>0.36$.
Applying amplitude amplification with $O(1)$ repetitions,
we can increase the amplitude of $|\hat{\x}\rangle$ close to 1 in $|\bar{\x}\rangle$.
Since
\begin{equation} \label{solution}
|\x\rangle = \frac{1}{\sqrt{m+1}} \sum_{l=0}^m |l\rangle |\x(t_l)\rangle
= \frac{1}{\sqrt{m+1}} \sum_{l=0}^m e^{2\pi i \lambda_j l \Delta t} |l\rangle |E_j\rangle.
\end{equation}
If we apply quantum Fourier inverse transform to the first register of $|\hat{\x}\rangle$, then we can obtain an $\epsilon$ approximation of
\begin{equation} \label{qft on solution}
\frac{1}{m+1}\sum_{k=0}^m \left(\sum_{l=0}^m
e^{2 \pi i l (\lambda_j  \Delta t-\frac{k}{m+1})} \right)
|k\rangle |E_j\rangle.
\end{equation}
Based on the analysis of QPE, we are left with a $k$ such that
$k/(m+1)$ is an $\epsilon$ approximation of $\lambda_j\Delta t$
with probability close to 1.

The complexity is dominated by the solving the linear system (\ref{coeff matrix}). As proved 
in Theorem 5 of \cite{berry2017quantum}, the condition number of $C_{m,k}$ is bounded by $O(\kappa(E)km)$. For the linear system (\ref{coeff matrix}), the dimension is $m(k+1)n$. Finally, based on the choices of $k$ and $m$, we obtain the claimed complexity of the above procedure.
\end{proof}

In a quantum computer, under certain conditions (e.g. block-encoding, sparse), the complexity to solve a linear system is linear at the condition number, logarithm on the precision, and the dimension. 
Thus, under these conditions, the complexity of Proposition \ref{main prop} can be simplified into $\widetilde{O}(\kappa(E)/\epsilon)$.

In the general case, the eigenvalue problem can be solved similarly. Before we state the algorithm, we remark that
in the proof of Proposition \ref{main prop}, we approximate $\lambda_j\Delta t$ up to error $\epsilon$.
This gives an $\epsilon/\Delta t = \rho\epsilon$ approximation of $\lambda_j$.
If we choose $\epsilon = \epsilon'/\rho$, then we obtain an $\epsilon'$-approximation of $\lambda_j$. The parameters $k, m$, and complexity should be changed accordingly.

\begin{breakablealgorithm}
\label{alg}
\caption{Quantum algorithm for computing the real eigenvalues}
\begin{algorithmic}[1]
\REQUIRE
(1). An $n\times n$ diagonalizable matrix $M$ that only with real eigenvalues.
Suppose the eigenvalues are $\{\lambda_1,\ldots,\lambda_n\}$ and the unit 
eigenvectors are $\{|E_1\rangle,\ldots,|E_n\rangle\}$.
\\
(2). A upper bound $\rho\geq 1$ of the eigenvalues. \\
(3). Quantum access to copies of the state $|\phi\rangle$, which formally equals $\sum_{j=1}^n \beta_j |E_j\rangle$.
\\
(4). The precision $\epsilon\in(0,1)$, $\Delta t=1/2\rho$, $m=\lceil \rho/\epsilon \rceil, k= \lceil\log(\rho/\epsilon) \rceil$.
\ENSURE The quantum state
\begin{equation} 
\sum_{j=1}^n \beta_j |\tilde{\lambda}_j\rangle |E_j\rangle 
\label{final state}
\end{equation} 
up to a normaliation, where $|\tilde{\lambda}_j- \lambda_j|\leq \epsilon$ for all $j$.
\STATE Construct the matrix $C_{m,k}$ based on equation (\ref{coeff matrix}).
\STATE Use quantum linear algebraic technique to create the state $C_{m,k}^{-1}|0..0\rangle|\phi\rangle$.
\STATE Apply quantum Fourier inverse transform to the first register of $C_{m,k}^{-1}|0..0\rangle|\phi\rangle$.
\STATE Perform measurements on the last register, return the post-selected state if the output is $\ket{0}$.
\end{algorithmic}
\end{breakablealgorithm}

\begin{thm}\label{thm2}
Let $M$ be an $n\times n$ diagonalizable matrix which only has real eigenvalues $\{\lambda_1,\ldots,$ $\lambda_n\}$. Assume that the corresponding unit eigenvectors are $\{|E_1\rangle,\ldots,|E_n\rangle\}$. Let $\rho\geq 1$ be a upper bound of the eigenvalues.
Suppose that $C_{m,k}^{-1}$ can be implemented in a quantum computer in time $T(\kappa(C_{m,k}),\epsilon,m(k+1)n)$ to precision $\epsilon$. 
Given access to copies of the state $|\phi\rangle = \sum_{j=1}^n \beta_j |E_j\rangle$, then with probability close to 1, Algorithm \ref{alg}
returns
\be \label{thm2:eq}
\frac{1}{Z}
\sum_{j=1}^n \beta_j |\tilde{\lambda}_j\rangle |E_j\rangle
\ee
in time
\be \label{thm2:complexity}
O\left(T_0\times  T\left(
\frac{\rho \kappa(E)\log(\rho /\epsilon)}{\epsilon},
\frac{\epsilon}{\rho},
\frac{n\rho\log(\rho /\epsilon)}{\epsilon}\right)
\times \frac{\rho\kappa(E)\sqrt{\log(\rho /\epsilon)}}{\epsilon}  \right),
\ee
where $Z$ is the normalization factor, $T_0$ is the complexity to prepare the initial state,
$\kappa(E)$ is the condition number of the matrix generated by 
the eigenvectors, and $|\lambda_j-\tilde{\lambda}_j|\leq  \epsilon$ for all $j$.
\end{thm}

\begin{proof}
By equation (\ref{eq:approsolu}), up to a normalization
\beas
C_{m,k}^{-1}|0..0\rangle |\phi\rangle &=&
\sum_{j=1}^n \beta_j C_{m,k}^{-1}|0..0\rangle |E_j\rangle \\
&=& \sum_{j=1}^n \beta_j \left(\sum_{p=0}^{m}  |p(k+1)\rangle |\x_{p,0}\rangle + \sum_{p=0}^{m-1} \sum_{q=1}^k |p(k+1)+q\rangle |\x_{p,q}\rangle \right).
\eeas
We will consider the normalization constant later.
Since $k+1$ is invertible modulo $m(k+1)+1$, we can find $(k+1)^{-1}$ and multiple it on the first register. As a result, we obtain
\be \label{thm2-proof:eq1}
\sum_{j=1}^n \beta_j \left(\sum_{p=0}^{m}  |p\rangle |\x_{p,0}\rangle |0\rangle + |0\rangle^\bot \right).
\ee
In the above state, we added a new ancilla qubit $|0\rangle$ to separate the two summation terms. This is feasible as the base states in the first register are orthogonal to each other.

By equation (\ref{eq solu}), Proposition \ref{prop:Solution error} and Lemma \ref{lemma2}, the state (\ref{thm2-proof:eq1}) is $\epsilon$-close to the state
\be \label{thm2-proof:eq2}
\sum_{j=1}^n\sum_{p=0}^{m} \beta_j e^{2\pi i \lambda_j p \Delta t }|p\rangle |E_j\rangle |0\rangle + |0\rangle^\bot
.
\ee
If we apply quantum Fourier inverse transform to $|p\rangle$,
then we obtain
\be \label{thm2-proof:eq3}
\frac{1}{\sqrt{m+1}}\sum_{j=1}^n\sum_{q=0}^{m}\sum_{p=0}^{m} \beta_j e^{2\pi i  p (\lambda_j\Delta t-\frac{q}{m+1}) }|q\rangle |E_j\rangle |0\rangle + |0\rangle^\bot
= \sqrt{m+1} \sum_{j=1}^n \beta_j| \tilde{\lambda}_j \Delta t \rangle |E_j\rangle |0\rangle + |0\rangle^\bot.
\ee
The above equality is caused by a similar reason to the QPE, see (\ref{eq-qpe}) in Appendix \ref{Quantum phase estimation}.

Note that the state (\ref{thm2-proof:eq2}) is not normalized.
By \cite[Lemma 3]{berry2017quantum}, $\|C^{-1}_{m,k}\| \leq 3\kappa(E)\sqrt{k}m$, As a result, $\|C^{-1}_{m,k}|0..0\rangle |\phi\rangle\| \leq \|C^{-1}_{m,k}\| \leq 3\kappa(E)\sqrt{k}m$. 
So the normalization constant is smaller than $3\kappa(E)\sqrt{k}m$.
Thus the amplitude of $|0\rangle$ in the state (\ref{thm2-proof:eq2}), which is also the amplitude of $|0\rangle$ in the state (\ref{thm2-proof:eq3}), is at least
\[
\frac{1}{3\kappa(E)\sqrt{k}m}
\sqrt{\sum_{p=0}^m \left\|\sum_{j=1}^n \beta_j e^{2\pi i \lambda_j p \Delta t } |E_j\rangle \right\|_2^2}.
\]
Note that if $p=0$, then $\|\sum_{j=1}^n \beta_j |E_j\rangle \|_2^2=\||\phi\rangle\|_2^2=1$.
Therefore, the amplitude of $|0\rangle$ of the state (\ref{thm2-proof:eq3}) is larger than $1/3\kappa(E)\sqrt{k}m$. By amplitude amplification and the choices of the parameters in the algorithm, the complexity to obtain the state (\ref{thm2:eq}) is
(\ref{thm2:complexity}) as claimed.
\end{proof}

\begin{rem}{\rm
There are several methods to find the upper bound of the eigenvalues. A simple one
is based on the Gershgorin circle theorem \cite{van1983matrix}, which states that every eigenvalue of 
$M$ lies within at least one of the Gershgorin discs $\{z\in \mathbb{C}: |z-m_{ii}|
\leq \sum_{j\neq i} |m_{ij}|\}$, where $i=1,\ldots,n$. As a result, $\|M\|_1,\|M\|_{\infty}$ are upper bounds of the
eigenvalues. The spectral norm $\|M\|$ also provides 
a upper bound of the eigenvalues. 
Since we can estimate the singular values efficiently in a quantum computer,
we can estimate the spectral norm  efficiently as well. For instance, see \cite{kerenidis2017quantum}.
}\end{rem}

\begin{rem}{\rm
For eigenvalue problems of diagonalizable matrices, Bauer and Fike in \cite{bauer1960norms}
shown that the condition number $\kappa(E)$
describes the stability and conditioning of calculating the eigenvalues.
If $\kappa(E)$ is large, then small permutations on the matrix $M$ will give rise to large permutations on the eigenvalues. This makes the calculation
of the eigenvalues inaccurate.
Thus the complexity of Algorithm
\ref{alg} depends on $\kappa(E)$ seems realistic. The dependence on $\kappa(E)$ might be
improved by using the variable time amplitude amplification technique \cite{ambainis:hal-00678197}.
}\end{rem}

\subsection{Sparse matrices}

In the following, we consider sparse matrices in which the complexity of inverting $C_{m,k}$ can be determined explicitly.
The following lemma about the sparsity of $C_{m,k}$ is easy to prove.

\begin{lem}
\label{sparsity}
If $M$ is $s$-sparse, then the sparsity of $C_{m,k}$ is $\Theta(s+k)$.
\end{lem}

\begin{proof}
By equation (\ref{coeff matrix}), for any $p\in\{0,\ldots,m\}$, $q\in\{0,\ldots,k\}$ and $r\in\{1,\ldots,n\}$
\beas
\langle p(k+1)+q|\otimes \langle r| C_{m,k}
&=& \langle p(k+1)+q|\otimes \langle r| - \langle p(k+1)+q-1|\otimes \langle r|\frac{2\pi i M \Delta t}{q} \\
&& -\, \delta_q^p \sum_{q'=0}^k \langle p(k+1)+q'|\otimes \langle r|.
\eeas
The nonzero element of this row vector 
is bounded by $1+s+k+1=\Theta(s+k)$.
Similar analysis also holds for the columns of $C_{m,k}$.
\end{proof}

By the quantum linear solver \cite[Theorem 5]{childs2017quantum}, if $M$ is $s$ sparse, then combining Lemma \ref{sparsity},
\beas  
& & T(\kappa(C_{m,k}),\epsilon,m(k+1)n) \\
&=& O((s+\log(1/\epsilon))\kappa(C_{m,k})  (\log n\epsilon^{-1}\log(1/\epsilon)) {\rm poly}{\log ((s+\log(1/\epsilon))\kappa(C_{m,k})/\epsilon})) \\
&=& \widetilde{O}(s\kappa(C_{m,k})).
\eeas
By \cite[Theorem 5]{berry2017quantum},
$\kappa(C_{m,k})=O(\kappa(E)km)=O(\kappa(E)\epsilon^{-1}(\log 1/\epsilon))$,
Thus, equation (\ref{thm2:complexity}) can be simplified 
into
$
\widetilde{O}(s\rho^2\kappa(E)^2/\epsilon^2).
$

The quantum linear solver \cite{childs2017quantum}
depends on an oracle $\mathcal{O}_C$ to query $C_{m,k}$.
Assume that we have the oracle  $\mathcal{O}_M$ to query $M$.
It is defined as
\be \label{eq:oracle}
\mathcal{O}_M|i,j,z\rangle = 
|i,j,z \oplus m_{ij}\rangle, \quad
\mathcal{O}_M|i,l\rangle = 
|i,\nu(i,l)\rangle,
\ee
where $\nu(i,l)$ is the index of
the $l$-th nonzero element in the
$i$-th row/column.
For any $0\leq p_1,p_2\leq m, 0\leq q_1,q_2\leq k$ and $1\leq r_1,r_2\leq n$, it is easy to check that
\bea
&& \langle p_1(k+1)+q_1,r_1|C_{m,k}|
p_2(k+1)+q_2,r_2 \rangle \\
&=& \delta^{r_1}_{r_2} \delta^{p_1}_{p_2} \delta^{q_1}_{q_2} 
-[p_2\leq m-1][q_2\leq k-1] \delta^{p_1}_{p_2} \delta^{q_1}_{q_2+1} 
\langle r_1 |\frac{2\pi i M\Delta t}{q_2+1}|r_2\rangle \label{fufu1} \\
&& -\, [p_2\leq m-1]\delta^{r_1}_{r_2} \delta^{p_1}_{p_2+1} \delta^{q_1}_0,
 \label{fufu2}
\eea
where the notation $[a\leq b]$ means that it equals 1 if $a\leq b$ and 0 otherwise. 
The three terms in equations (\ref{fufu1}), (\ref{fufu2}) cannot coexist, so it is
easy to build the oracle $\mathcal{O}_C$
to query $C_{m,k}$ when we have
$\mathcal{O}_M$.
The cost to build this oracle $\mathcal{O}_C$
is the same as that to build $\mathcal{O}_M$.


\begin{thm}
\label{cor of thm2}
With the same assumptions and notation as Theorem \ref{thm2}. Suppose there is an oracle defined by (\ref{eq:oracle}) to query $M$.
If $M$ is $s$ sparse, then
Algorithm \ref{alg} returns the state (\ref{thm2:eq}) in time 
$
\widetilde{O}(s\rho^2\kappa(E)^2/\epsilon^2).
$
\end{thm}

To obtain all the eigenvalues with classical output, we can perform measurements on the state (\ref{thm2:eq}).
The following lemma shows how many measurements we should perform.

\begin{lem}
\label{distribution}
Let $\mathcal{P}=\{p_1,\ldots,p_n\}$ be a  probability distribution.
Set $p_{\max}=\max\{p_1,\ldots,p_n\}$.
To see every events under the distribution $\mathcal{P}$
it suffices to  make
$O(p_{\max}^{-1} \log(n/\delta))$
measurements. The success probability is
at least $1-\delta$.
\end{lem}

We defer the proof of this lemma to Appendix \ref{uniform-distribution}. 
To apply Lemma \ref{distribution}, we need to analyze the probability to obtain each eigenvalue.
Assume that $M$ has $d$ distinct eigenvalues $\lambda_1,\ldots,\lambda_d$,
we can rewrite the initial state as $\sum_{j=1}^d \gamma_j |V_j\rangle$, where $|V_j\rangle$ is a normalized vector generated by the eigenvectors corresponding to $\lambda_j$. Then the right hand side of equation (\ref{thm2:eq}) can be written as  $\frac{1}{Z}\sum_{j=1}^d \gamma_j |\tilde{\lambda}_j\rangle|V_j\rangle$. The probability to obtain $\tilde{\lambda}_j$ equals $p_j = {|\gamma_j|^2}/{Z^2}$. Note that 
$Z^2 = \sum_{j=1}^d |\gamma_j|^2$, so $p_{\max} = \max_j {|\gamma_j|^2}/{Z^2} \geq 1/d$. Thus, it suffices  to make $O(d\log d)=O(n\log n)$ measurements by Lemma \ref{distribution}. 
The sparse access oracle of $M$ can be built in time $O(sn)$. 
Thus we obtain the following result.

\begin{cor}
\label{cor2 of thm2}
With the same assumptions and notation as Theorem \ref{cor of thm2},
then there is a quantum algorithm that returns all the eigenvalues of $M$ up to precision $\epsilon$ in time
$
\widetilde{O}(sn\rho^2\kappa(E)^2/\epsilon^2).
$
\end{cor}

\subsection{What if the matrices have complex eigenvalues}

Algorithm \ref{alg} also works for diagonalizable matrices that only have
purely imaginary eigenvalues. It suffices to
change $2\pi iM$ into $2\pi M$ in the ODE (\ref{eq1:new}). 
In the following, we consider the problem that what would
happen if we apply Algorithm \ref{alg} directly to any diagonalizable matrix. 
We will show that Algorithm \ref{alg} has many difficulties in estimating the complex eigenvalues, even the real parts.

To do the analysis, we just focus on the special case that the initial state is  $|E_j\rangle$.
Denote the corresponding eigenvalue 
as $\lambda_{j0}+i\lambda_{j1}$, where $\lambda_{j0},\lambda_{j1}\in \mathbb{R}$.
Then similar to the analysis of obtaining (\ref{solution}), we 
will obtain an approximation of
\begin{equation} \label{solution-complex}
|\x\rangle = \frac{1}{Z} \sum_{l=0}^m |l\rangle |\x(t_l)\rangle
= \frac{1}{Z} \sum_{l=0}^m e^{2\pi i \lambda_j l \Delta t} |l\rangle |E_j\rangle,
\end{equation}
where $Z$ is the normalization constant. Since $\lambda_j\in \mathbb{C}$, usually $Z\neq \sqrt{m+1}$. If we set 
$\lambda_{j1} \Delta t = -b$ for simplicity,
then
\[
Z^2 = \sum_{l=0}^m e^{4\pi b l}
=\frac{1-e^{4\pi b (m+1)}}{1-e^{4\pi b}}.
\]

To make sure the quantum differential equation solver works,
one assumption made in \cite{berry2017quantum} is that the real parts of the eigenvalues are non-positive. 
This assumption relates to the stability of the differential equation.
Here we should make
the  same assumption, that is $\lambda_{j1}>0$ as we use $2\pi i M$ in the differential equation (\ref{eq1:new}). Thus $b<0$.
When concerning about 
computing eigenvalues, this assumption is easy to be satisfied.
We just need to consider $M-i \rho I$. Also notice that $\lambda_{j1} \Delta t = -b$, so $-1\leq b <0$.

If we apply quantum Fourier inverse transform to the first register of $|\x\rangle$, then
we obtain
\be\label{state:complex-eigenvalues}
\frac{1}{Z\sqrt{m+1}} \sum_{k=0}^{m} \sum_{l=0}^{m}
e^{2\pi i l(\lambda_{j0} \Delta t -\frac{k}{m+1})} 
e^{2\pi lb} |k\rangle|E_j\rangle.
\ee
Assume that $q$ is the integer such that $\lambda_{j0} \Delta t-\frac{q}{m+1}>0$ is minimal. For convenience, we set $a = \lambda_{j0} \Delta t -\frac{q}{m+1}$, then $a \leq 1/(m+1)$. The probability of $ k = q + s({\rm mod}~m+1)$ is
\bea
P_s &=& \frac{1}{Z^2(m+1)}
\left| \sum_{l=0}^{m} e^{2\pi i l (a - \frac{s}{m+1})} e^{2\pi lb} \right|^2 \\
&=& \frac{(1-e^{4\pi b})}{(1-e^{4\pi b (m+1)})(m+1)}
\times
\frac{(e^{2\pi b(m+1)}-1)^2 + 4 e^{2\pi  b(m+1)} \sin^2(\pi  (m+1) (a-\frac{s}{m+1}))}{(e^{2\pi  b}-1)^2 + 4 e^{2\pi  b} \sin^2(\pi (a-\frac{s}{m+1}))} \\
&\leq& \frac{(1-e^{4\pi b})(e^{2\pi b(m+1)}+1)^2}
{4 e^{2\pi  b}(1-e^{4\pi b (m+1)})(m+1) \sin^2(\pi (a-\frac{s}{m+1}))} \\
&\leq& 
\frac{e^{2\pi}}{16} \times 
\frac{(1-e^{4\pi b}) (m+1)}{(1-e^{4\pi b (m+1)})} \times
\frac{ 1} {(a(m+1)-s)^2}.
\eea

Next, we aim to bound the probability of
obtaining an integer $k$ such that $|k - q| > r$, 
where $r$ is a positive integer characterizing
the desired tolerance to error. 
For simplicity, we assume that $m+1 = 2^t$.
This probability is
\bea
{\rm Pr}(|k - q| > r) &=& \left(\sum_{s= -2^{t-1} +1}^{- (s+1)} +\sum_{r+1}^{2^{t-1}} \right) P_s \\
&\leq& \frac{e^{2\pi}}{16} \times 
\frac{(1-e^{4\pi b}) (m+1)}{(1-e^{4\pi b (m+1)})} \times
\left(\sum_{s= -2^{t-1} +1}^{- (s+1)} +\sum_{s = r+1}^{2^{t-1}} \right) \frac{ 1} {(a(m+1)-s)^2} \\
&\leq& \frac{e^{2\pi}}{16} \times 
\frac{(1-e^{4\pi b}) (m+1)}{(1-e^{4\pi b (m+1)})} \times
\left(\sum_{s= -2^{t-1} +1}^{- (s+1)} \frac{1}{s^2} +\sum_{s = r+1}^{2^{t-1}} \frac{1}{(s-1)^2}\right) \\
&\leq& \frac{e^{2\pi}}{8} \times 
\frac{(1-e^{4\pi b}) (m+1)}{(1-e^{4\pi b (m+1)})} \times
\sum_{s = r}^{2^{t-1} - 1} \frac{1}{s^2} \\
&\leq& \frac{e^{2\pi}}{8} \times 
\frac{(1-e^{4\pi b}) (m+1)}{(1-e^{4\pi b (m+1)})} \times
\int_{r-1}^{2^{t-1} - 1} \frac{1}{s^2} ds \\
&\leq&  \frac{e^{2\pi}}{8(r-1)} \times 
\frac{(1-e^{4\pi b}) (m+1)}{(1-e^{4\pi b (m+1)})} .
\eea
Assume that $b = - C/(m+1)$ for some $C \in\{1,2,\ldots,m+1\}$, then
\be
{\rm Pr}(|k - q| > r) \leq \frac{e^{2\pi}}{8(r-1)} \times \frac{4\pi C}{1-e^{ - 4\pi }} = \Theta(C/r) .
\ee

Suppose we wish to approximate $\lambda_{j0}\Delta t$ to an accuracy $2^{-x}$, that is, we choose $r = 2^{t-x} - 1$. By making $t = x + y$, the probability of obtaining an approximation correct to this accuracy is at least $1 - C/2^y$.

Since $b = -\lambda_{j1} \Delta t$, we have
$\Delta t = C/(m+1)\lambda_{j1}$. 
Now let $k$ be the integer such that $k/(m+1)$
is an $2^{-x}$-approximation of $\lambda_{j0} \Delta t$.
Then from $m + 1=2^{x+y}$ and
\[
\Big|\lambda_{j0} \Delta t -\frac{k}{m+1}\Big|=\Big|\lambda_{j0} \frac{C}{(m+1)\lambda_{j1}} -\frac{k}{m+1}\Big|\leq \frac{1}{2^x},
\]
we obtain
\be \label{fufu3}
\Big|\lambda_{j0} - \frac{k\lambda_{j1}}{C}\Big|\leq \frac{2^y\lambda_{j1}}{C}.
\ee

This leads to a contradiction to the choices of $C$. We cannot find a $C$ such that the success probability is high, meanwhile, the error is small.
The above analysis shows that when the
eigenvalues are complex, Algorithm \ref{alg} may not return a good approximation of the real parts of the eigenvalues.

\begin{rem}
\label{remark}
{\rm 
To compute all the complex eigenvalues, it suffices to have a quantum algorithm to compute all the real parts. More precisely,
let $M$ be a diagonalizable matrix with complex eigenvalues $\lambda_j+i\mu_j$, where $j=1,\ldots,n$. The corresponding unit
eigenvalues are $|E_1\rangle,\ldots,|E_n\rangle$.
Suppose that we have a quantum algorithm that can output
$
\sum_{j=1}^n \beta_j |\lambda_j\rangle|E_j\rangle
$  (up to a normalization)
when the input is
$
\sum_{j=1}^n \beta_j |E_j\rangle.
$
Then we can apply this algorithm further to obtain
$
\sum_{j=1}^n \beta_j |\lambda_j\rangle|\mu_j\rangle|E_j\rangle.
$
This is obtained by considering $iM$ and viewing $
\sum_{j=1}^n \beta_j |\lambda_j\rangle|E_j\rangle
$ as the new initial state.
}\end{rem}

\section{A quantum algorithm to estimate complex eigenvalues}
\label{sec:A quantum algorithm to estimate complex eigenvalues}

Based on Remark \ref{remark}, 
it suffices to propose a quantum
algorithm that can estimate the real parts
of the eigenvalues.
The following is a  simple idea but with an extra assumption to
generalize Algorithm \ref{alg} to 
achieve this goal.

Suppose the complex eigenvalues 
of $M$ are
$\lambda_1+i \mu_1, \ldots, \lambda_n+i \mu_n$. The corresponding eigenvectors
are $|E_1\rangle,\ldots,|E_n\rangle$.
Denote $\overline{M}$ as
the complex conjugate of $M$. Then 
$\{(\lambda_j-i\mu_j,|\overline{E}_j\rangle):j=1,2,\ldots,n\}$
are the eigenpairs of $\overline{M}$.
Here $|\overline{E}_j\rangle$ is the complex conjugate of 
$|E_j\rangle$.

Consider the following differential equation:
\be \label{special ODE}
\begin{cases} 
\ds \frac{d\x}{dt} = \pi i (M\otimes I + I \otimes \overline{M}) \x, \\
\x(0) = |E_j\rangle\otimes |\overline{E}_j\rangle.
\end{cases}
\ee
It's obvious that the solution is
\bea
|\x(t)\rangle = e^{\pi i t(M\otimes I + I \otimes \overline{M})} |E_j\rangle\otimes |\overline{E}_j\rangle = e^{\pi i t M}|E_j\rangle \otimes e^{\pi i t \overline{M}} |\overline{E}_j\rangle 
= e^{2\pi i t \lambda_j}|E_j\rangle \otimes |\overline{E}_j\rangle.  
\label{solution of the special ODE}
\eea
Similar to the differential equation (\ref{eq1:new}), if we use quantum 
algorithm to solve (\ref{special ODE}), then we obtain
a superposition of the solutions (see
equation (\ref{solution})):
\be
\frac{1}{\sqrt{m+1}}\sum_{l=0}^{m} |l\rangle |\x(t_l) \rangle
=\frac{1}{\sqrt{m+1}}\sum_{l=0}^{m} e^{2\pi i l \lambda_j \Delta t} |l\rangle |E_j\rangle \otimes |\overline{E}_j\rangle.
\ee
If we apply quantum Fourier inverse transform 
to the first register
of the superposition, then we obtain
\be \label{apply QFT to the solution}
\frac{1}{m+1}\sum_{k=0}^{m} \left(\sum_{l=0}^{m-1} e^{2\pi i l (\lambda_j \Delta t-\frac{k}{m+1})}\right) |k\rangle |E_j\rangle \otimes |\overline{E}_j\rangle
\approx |\tilde{\lambda}_j \Delta t\rangle |E_j\rangle \otimes |\overline{E}_j\rangle.
\ee

In the general case, we can choose the 
initial vector of the differential equation
(\ref{special ODE}) as
\be \label{initial state}
\x(0) = \sum_{j=1}^n \beta_j |E_j\rangle \otimes |\overline{E}_j\rangle.
\ee
Then the superposition of the
solutions is proportional to
\be
\sum_{l=0}^{m-1}\sum_{j=1}^n \beta_j e^{2\pi i l \lambda_j \Delta t} |l\rangle|E_j\rangle \otimes |\overline{E}_j\rangle.
\ee
Similarly, we can apply quantum Fourier 
inverse transform 
to the first register to estimate the
real parts of the eigenvalues
\be
\label{transformation 1}
\sum_{k=0}^{m-1}\sum_{j=1}^n \beta_j \left(\sum_{l=0}^{m-1}e^{2\pi i l (\lambda_j \Delta t-\frac{k}{m})}\right) |k\rangle|E_j\rangle \otimes |\overline{E}_j\rangle
\approx \sum_{j=1}^n \beta_j
|\tilde{\lambda}_j \Delta t\rangle|E_j\rangle \otimes |\overline{E}_j\rangle.
\ee

If we further view  (\ref{transformation 1}) as the initial state 
and implement the above procedure by changing
$\pi i (M\otimes I+I\otimes \overline{M})$
into $\pi (M\otimes I-I\otimes \overline{M})$
in the differential equation (\ref{special ODE}).
Then we can further obtain the state 
\be
\sum_{j=1}^n \beta_j
|\tilde{\lambda}_j \Delta t\rangle
|\tilde{\mu}_j \Delta t\rangle
|E_j\rangle \otimes |\overline{E}_j\rangle
\ee
up to a normalization. 
With the above idea, the quantum algorithm to estimate complex eigenvalues
can be stated as follows.

\begin{breakablealgorithm}
\label{alg-new}
\caption{Quantum algorithm for computing the eigenvalues of diagonalizable matrices}
\begin{algorithmic}[1]
\REQUIRE
(1). An $n\times n$ diagonalizable matrix $M$.
Suppose that the eigenvalues are $\{\lambda_1+i \mu_1, \ldots, \lambda_n+i \mu_n\}$ and the unit 
eigenvectors are $\{|E_1\rangle,\ldots,|E_n\rangle\}$.
\\
(2). A upper bound $\rho\geq 1$ of the eigenvalues. \\
(3). Quantum access to copies of the state $|\phi\rangle$ of the form $|\phi\rangle=\sum_{j=1}^n \beta_j |E_j\rangle |\overline{E}_j\rangle$. 
\\
(4). The precision $\epsilon\in(0,1)$, $\Delta t=1/\rho$, $m=\lceil \rho/\epsilon \rceil, k=2\lceil\log(\rho/\epsilon)\rceil$.
\ENSURE The quantum state
\be 
\sum_{j=1}^n \beta_j
|\tilde{\lambda}_j \rangle
|\tilde{\mu}_j \rangle
|E_j\rangle  |\overline{E}_j\rangle
\ee
up to a normalization, where $|\tilde{\lambda}_j- \lambda_j|\leq \epsilon$ and
$|\tilde{\mu}_j- \mu_j|\leq \epsilon$ for all $j$.
\STATE Construct matrices $C_1:=C_{m,k}(\pi i \Delta t(M\otimes I+I\otimes \overline{M}))$ and
$C_2:=C_{m,k}(\pi  \Delta t (M\otimes I-I\otimes \overline{M}))$ based on equation (\ref{coeff matrix}).
\STATE Use quantum linear algebraic technique to construct $|\phi_1\rangle=C_1^{-1}|0..0\rangle|\phi\rangle$.
\STATE Apply quantum Fourier inverse transform to the first register of $|\phi_1\rangle$. Then the result is proportional to
$\sum_j \beta_j
|\tilde{\lambda}_j \rangle
|E_j\rangle  |\overline{E}_j\rangle|0\rangle + |0\rangle^\bot$. Define
$|\phi_2\rangle=\sum_j \beta_j
|\tilde{\lambda}_j \rangle|0..0\rangle
|E_j\rangle  |\overline{E}_j\rangle|0\rangle + |0\rangle^\bot$.
\STATE Use quantum linear algebraic technique to construct $|\phi_3\rangle= I\otimes C_2^{-1} |\phi_2\rangle$.
\STATE Apply quantum Fourier inverse transform to the second register of $|\phi_3\rangle$.
\STATE Measure the last register of $|\phi_3\rangle$ and return the post-selected state if the output is  $|0\rangle$.
\end{algorithmic}
\end{breakablealgorithm}

\begin{lem}
\label{condition number}
Let $P$ be an nonsingular matrix, 
then the condition number $\kappa(P\otimes \overline{P})
=\kappa(P)^2$.
\end{lem}

\begin{proof}
Assume that the SVD of $P$ is $UDV$, 
then the SVD of $\overline{P}$ is $\overline{U}D\overline{V}$. Thus
$P,\overline{P}$ have the same singular values.
If $\sigma_1\geq \sigma_2\geq \cdots 
\geq \sigma_n>0$ are the singular values of $P$,
then the singular values of $P\otimes \overline{P}$
are $\sigma_i \sigma_j$. Therefore
\[
\kappa(P\otimes \overline{P})
=\frac{\max_{i,j}{\sigma_i \sigma_j}}{\min_{i,j}{\sigma_i \sigma_j}}
=\frac{\sigma_1^2}{\sigma_n^2}
=\kappa(P)^2
\]
as claimed.
\end{proof}

Note that if $M=PDP^{-1}$, then 
$M\otimes I \pm I \otimes \overline{M}
=(P\otimes \overline{P})(D\times I \pm
I \otimes \overline{D}) (P\otimes \overline{P})^{-1}$. So the matrix
of the eigenvectors of $M\otimes I \pm I \otimes \overline{M}$ is $P\otimes \overline{P}$.
The cost of Algorithm \ref{alg-new} is mainly determined
by the quantum linear algebraic technique to
implement $C_1^{-1},C_2^{-1}$. 
Similar to the proof of Theorem \ref{thm2},
we have

\begin{thm} \label{main theorem}
Let $M$ be an $n\times n$ diagonalizable matrix. Assume that its eigenvalues are
$\lambda_1+i \mu_1, \ldots, \lambda_n+i \mu_n$,
and the corresponding normalized eigenvectors are
$|E_1\rangle,\ldots,|E_n\rangle$.
Let $\rho\geq 1$ be a upper bound of the eigenvalues.
Given access to copies of the state
$\sum_{j=1}^n\beta_j|E_j\rangle|\overline{E}_j\rangle$, then
Algorithm \ref{alg-new} returns
\be
\sum_{j=1}^n \beta_j |\tilde{\lambda}_j\rangle |\tilde{\mu}_j\rangle |E_j\rangle|\overline{E}_j\rangle
\ee
up to a normalization in time
\be \label{main thm:complexity}
O\left(T_0\times
T\left(
\frac{\rho\kappa(E)^2\log(\rho/\epsilon)}{\epsilon}, \frac{\epsilon}{\rho},
\frac{n\rho\log(\rho/\epsilon)}{\epsilon}
\right)
\times \frac{\rho^2\kappa(E)^4}{\epsilon^2} \log \frac{\rho}{\epsilon}   \right),
\ee
where $T_0$ is the complexity to generate the initial state, $\kappa(E)$ is the condition number of the matrix generated by 
the eigenvectors, $|\lambda_j-\tilde{\lambda}_j| \leq \epsilon$
and $|\mu_j-\tilde{\mu}_j| \leq \epsilon$ for all $j$. Especially, if $M$ is $s$ sparse, then the complexity is 
$
\widetilde{O}(T_0s\rho^3\kappa(E)^6(\log n)/\epsilon^3).
$
\end{thm}

\begin{proof}
The proof of this theorem is similar to that of Theorem \ref{thm2}.
By equation (\ref{eq:approsolu}), up to a normalization
\beas
C_1^{-1}|0..0\rangle |\x(0)\rangle &=&
 \sum_{j=1}^n \beta_j C_1^{-1}|0..0\rangle |E_j\rangle |\overline{E}_j\rangle\\
&=&  \sum_{j=1}^n \beta_j \left(\sum_{p=0}^{m}  |p(k+1)\rangle |\x_{p,0}\rangle + \sum_{p=0}^{m-1} \sum_{q=1}^k |p(k+1)+q\rangle |\x_{p,q}\rangle \right).
\eeas
Since $k+1$ is invertible modulo $m(k+1)+1$, we can find $(k+1)^{-1}$ and multiple it on the first register. As a result, we obtain
\be \label{mainthm2-proof:eq1}
\sum_{j=1}^n \beta_j \left(\sum_{p=0}^{m}  |p\rangle |\x_{p,0}\rangle |0\rangle + |0\rangle^\bot \right)
\ee
by adding a new ancilla qubit $|0\rangle$ to separate the two summation terms. 

By Proposition \ref{prop:Solution error} and Lemma \ref{lemma2}, the state
(\ref{mainthm2-proof:eq1}) 
is close to the state
\be \label{mainthm2-proof:eq2}
\sum_{j=1}^n\sum_{p=0}^{m} \beta_j e^{2\pi i \lambda_j p \Delta t }|p\rangle |E_j\rangle|\overline{E}_j\rangle |0\rangle + |0\rangle^\bot
.
\ee
If we apply quantum Fourier inverse transform to $\ket{p}$,
then we obtain
\be \label{mainthm2-proof:eq3}
\sum_{j=1}^n\sum_{q=0}^{m}\sum_{p=0}^{m} \beta_j e^{2\pi i  p (\lambda_j\Delta t-\frac{q}{m+1}) }|q\rangle |E_j\rangle|\overline{E}_j\rangle |0\rangle + |0\rangle^\bot
\approx  \sum_{j=1}^n \beta_j |\tilde{\lambda}_j\Delta t \rangle |E_j\rangle|\overline{E}_j\rangle |0\rangle + |0\rangle^\bot
\ee
up to a normalization, where $|\tilde{\lambda}_j -\lambda_j | \leq \epsilon$.

To estimate the imaginary parts of the eigenvalues,
we can consider the differential equation
\be 
\ds \frac{d\x(t)}{dt} = \pi  (M\otimes I - I \otimes \overline{M}) \x(t)
\ee
with the initial state (\ref{mainthm2-proof:eq3}) in that
$e^{\pi t (M\otimes I - I \otimes \overline{M})}
|E_j\rangle|\overline{E}_j\rangle
=e^{2 \pi i t \mu_j}|E_j\rangle|\overline{E}_j\rangle.$
To obtain the superposition of this differential equation,
we apply the inverse of $C_2=C_{m,k}(\pi t (M\otimes I - I \otimes \overline{M}))$ to the state (\ref{mainthm2-proof:eq3})
to obtain
\be \label{mainthm2-proof:eq4}
\sum_{j=1}^n \beta_j |\tilde{\lambda}_j\Delta t \rangle
C_2^{-1}|0..0\rangle |E_j\rangle|\overline{E}_j\rangle |0\rangle + |0\rangle^\bot.
\ee
Similar to the analysis of the state (\ref{mainthm2-proof:eq2}),
equation (\ref{mainthm2-proof:eq4}) is close to
\be \label{mainthm2-proof:eq5}
\sum_{j=1}^n\sum_{p=0}^{m} \beta_j e^{2\pi i \mu_j p \Delta t }|\tilde{\lambda}_j\Delta t \rangle|p\rangle |E_j\rangle|\overline{E}_j\rangle |0\rangle 
+ |0\rangle^\bot.
\ee
Apply quantum Fourier inverse transform to the second register, then we obtain
\be \label{mainthm2-proof:eq6}
\sum_{j=1}^n \beta_j |\tilde{\lambda}_j\Delta t \rangle|\tilde{\mu}_j\Delta t \rangle |E_j\rangle|\overline{E}_j\rangle |0\rangle 
+ |0\rangle^\bot,
\ee
 where $|\tilde{\mu}_j -\mu_j | \leq \epsilon$.

Next, we estimate the success probability.
So we need to compute the amplitude 
of $|0\rangle$ in the state (\ref{mainthm2-proof:eq6}).
By Lemma 3 of \cite{berry2017quantum} and Lemma \ref{condition number}, $\|C^{-1}_1\| \leq 3\kappa(E)^2\sqrt{k}m, \|C^{-1}_2\| \leq 3\kappa^2\sqrt{k}m$, As a result, $\|C^{-1}_1|0..0\rangle |x(0)\rangle\| \leq \|C^{-1}_1\| \leq 3\kappa(E)^2\sqrt{k}m$. Before normalization, the amplitude of $|0\rangle$ in the state  (\ref{mainthm2-proof:eq3}) equals
\[
\sqrt{\sum_{p=0}^m \left\|\sum_{j=1}^n \beta_j e^{2\pi i \lambda_j p \Delta t } |E_j\rangle \right\|_2^2} \geq 1
\]
in that when $p=0$, we have $\|\sum_{j=1}^n \beta_j |E_j\rangle \|_2^2=\||\phi\rangle\|_2^2=1$.
Therefore, the amplitude of $|0\rangle$ of the state (\ref{mainthm2-proof:eq3}) is larger than $1/3\kappa(E)^2\sqrt{k}m$. 
Similar analysis shows that
the amplitude of $|0\rangle$ of the state (\ref{mainthm2-proof:eq6}) is larger than $1/9\kappa(E)^4km^2$. 

If $M$ is $s$ sparse, then $M\otimes I \pm I\otimes \overline{M}$
is at most $2s$ sparse. By Lemma \ref{sparsity},
$C_1,C_2$ is $O(s+\log\frac{1}{\epsilon})$ sparse. By \cite[Theorem 5]{childs2017quantum},
the cost to solve an $s$ sparse linear system with condition number $\kappa$ in a quantum computer is $\widetilde{O}(s\kappa)$.
Thus, the claimed result comes from the cost of amplitude amplification and the choices of the parameters in the algorithm.
\end{proof}



\section{Normal matrices}
\label{Normal Matrices}

A square matrix is called normal if it commutes with its conjugate transpose. 
The condition of normality may be strong, but it includes the unitary, Hermitian, skew-Hermitian matrices and their real counterparts as special cases.
These matrices are of great interests to physicists \cite{macklin1984normal}.
In \cite{grone1987normal}, Grone et al. listed 70 different
equivalent conditions of normal matrices.
19 more were added later in \cite{elsner1998normal}.
One interesting result we will use in this paper is the fact that normal matrices can be diagonalized by unitary matrices.
The list \cite{grone1987normal,elsner1998normal} 
reflects the fact that normality arises in many ways.

In this section, we solve Problem \ref{main problem}
for normal matrices. A quantum algorithm
based on quantum singular value decomposition and quantum phase estimation will be given.

\subsection{Main result}

Assume that $M$ is
an $n\times n$ normal matrix, 
then 
there is a diagonal matrix
$\Lambda = {\rm diag}(\lambda_1,\ldots,\lambda_n)$ and a unitary matrix $U$ such that $M = U \Lambda U^\dag$.
Thus $\{ \sigma_j := |\lambda_j|: j=1,\ldots,n\}$ are the singular values of $M$.
Denote $\lambda_j = \sigma_j e^{2 \pi i \theta_j}$, the problem we want to solve in this
section is to estimate $\sigma_j, \theta_j$ for all $j$ up to certain precision in a quantum computer. 
To be more precise, assume that the $j$-th column of $U$ is $|u_j\rangle$. 
Let $|b\rangle$ be any given state. Since $U$ in unitary, there exist $\beta_1,\ldots,\beta_n$ such that
$|b\rangle=\sum_{j=1}^n \beta_j |u_j\rangle$. The main objective we want to obtain is 
to find a quantum algorithm to return
\begin{equation} \label{Quantum:EVD}
\sum_{j=1}^n \beta_j |u_j\rangle |\tilde{\sigma}_j\rangle |\tilde{\theta}_j\rangle,
\end{equation}
where $\tilde{\sigma}_j$ and $\tilde{\theta}_j$ are respectively the approximations
of $\sigma_j$ and $\theta_j$.

To state our main result, we need the technique of quantum singular value estimation.
Assume that the
singular value decomposition (SVD) of $M=\sum_{j=1}^n \sigma_j |u_j\rangle \langle v_j|$. Denote
\be
\widetilde{M} = \left(\begin{array}{cc} 
0 & M  \\
M^\dag & 0
\end{array}\right),
\ee
which is Hermitian.
The eigenvalues of $\widetilde{M}$ are $\{\pm \sigma_j: j=1,\ldots,n\}$.
The corresponding eigenvectors are
\begin{equation}
 |w_j^\pm\rangle
:=\frac{1}{\sqrt{2}} (|0\rangle|u_j\rangle \pm |1\rangle|v_j\rangle),
\quad
j=1,\ldots,n.
\end{equation}
Based on quantum phase estimation, we can implement the following transformation
\begin{equation} 
\label{Quantum:SVD1}
\sum_{j=1}^n \beta_j^+ |w_j^+\rangle + \beta_j^- |w_j^-\rangle
\mapsto \sum_{j=1}^n \beta_j^+ |w_j^+\rangle |\tilde{\sigma}_j\rangle + \beta_j^- |w_j^-\rangle |-\tilde{\sigma}_j\rangle
\end{equation}
in a quantum computer, where $|\sigma_j-\tilde{\sigma}_j|\leq \epsilon$. More details can be found in Appendix \ref{Quantum phase estimation}.
Usually, the complexity is poly-log at $n$ and linear at $1/\epsilon$. In the following, we shall use $O(T/\epsilon)$ to denote the complexity of implementing the transformation (\ref{Quantum:SVD1}).

With the above preliminaries,
our main result is stated as follows.
We will prove it in the next subsection.

\begin{thm}
\label{thm1}
Let $M$ be an $n$-by-$n$ normal matrix. 
Assume that its eigenvalue decomposition is
$\sum_{j=1}^n \sigma_j e^{2 \pi i\theta_j} |u_j\rangle \langle u_j|$. 
Suppose the
complexity to implement (\ref{Quantum:SVD1}) for $M$ to precision $\epsilon_1$ is $O(T/\epsilon_1)$, then there is a quantum algorithm that returns
the state (\ref{Quantum:EVD}) in time
$O(T/\epsilon_1\epsilon_2)$,
where $|\sigma_j - \tilde{\sigma}_j| \leq \epsilon_1$,
$|\theta_j - \tilde{\theta}_j| \leq \epsilon_2$ for all $j$.

\end{thm}

Let $M=(m_{ij})_{n\times n}$ be a matrix, we use $\|M\|_{\max}$ to denote the quantity $\max_{i,j}|m_{ij}|$.
As a direct application of   Theorem \ref{thm1} and Proposition \ref{prop:qsvd} in Appendix \ref{Quantum phase estimation},
we have the following result.

\begin{thm}
\label{cor1}
Let $M$ be an $s$ sparse $n\times n$ normal matrix, let $\rho$ be a upper bound of its eigenvalues.
Then there is a quantum algorithm that prepares the state (\ref{Quantum:EVD}) in time $\widetilde{O}(s\rho\|M\|_{\max}/\epsilon_1\epsilon_2)$, where $|\sigma_j - \tilde{\sigma}_j| \leq \epsilon_1$,
$|\theta_j - \tilde{\theta}_j| \leq \epsilon_2$ for all $j$.
\end{thm} 

If we perform measurements on the last two registers of (\ref{Quantum:EVD}), then 
we can obtain all the eigenvalues of $M$.
For sparse matrices, we have the following result.

\begin{cor}
Let $M$ be an $s$ sparse $n\times n$ normal matrix with eigenvalues $\sigma_je^{i\theta_j},~j=1,\ldots,n$. Let $\epsilon_1,\epsilon_2$ be the precisions to approximate  $\sigma_j,\theta_j$ respectively.
Assume that $\sigma_j\leq \rho$ for all $j$.
Then there is a quantum algorithm that returns all the eigenvalues in time
$
\widetilde{O}(sn\rho\|M\|_{\max}/\epsilon_1\epsilon_2).
$
\end{cor}

\begin{proof}
We choose the initial state as maximally mixed state
$
|\phi\rangle = \frac{1}{\sqrt{n}} \sum_{j=1}^n |j,j\rangle.
$
Assume that the eigenvalue decomposition of $M$ is
$\sum_{j=1}^n \sigma_j e^{2 \pi i\theta_j} |u_j\rangle \langle u_j|$. Then we have
$
|\phi\rangle = \frac{1}{\sqrt{n}} \sum_{j=1}^n |u_j,\bar{u}_j\rangle.
$
By Corollary \ref{cor1}, there is a quantum
algorithm that returns
$
\frac{1}{\sqrt{n}} \sum_{j=1}^n |u_j,\bar{u}_j\rangle |\tilde{\sigma}_j\rangle |\tilde{\theta}_j\rangle 
$
in time $\widetilde{O}(s\rho\|M\|_{\max}/\epsilon_1\epsilon_2)$, where  $|\sigma_j - \tilde{\sigma}_j| \leq \epsilon_1$,
$|\theta_j - \tilde{\theta}_j| \leq \epsilon_2$. By Lemma \ref{distribution},
to obtain all the eigenvalues, it suffices to make $O(n\log n)$ measurements.
\end{proof}

\subsection{Proof of Theorem \ref{thm1}}

For any $j$, denote $|v_j\rangle = e^{-2 \pi i \theta_j} |u_j\rangle$,
then the SVD of $M$ is
$\sum_{j=1}^n  \sigma_j |u_j\rangle \langle v_j|$.
We can rewrite the initial state $|0\rangle|b\rangle$ as
$
|0\rangle|b\rangle = \sum_{j=1}^n \frac{\beta_j}{\sqrt{2}} ( |w_j^+\rangle + |w_j^-\rangle ).
$
The main technique is described in the following lemma.

\begin{claim}
\label{technical lemma}
With the same notation as Theorem \ref{thm1}, then there is a unitary operator $W$ that performs
\begin{equation} \label{eq5}
W: \quad \sum_{j=1}^n \beta_j |0\rangle |u_j\rangle
\mapsto \sum_{j=1}^n \beta_j e^{-2 \pi i \theta_j} |0\rangle |u_j\rangle.
\end{equation}
The operator $W$ is implemented in time $O(T/\epsilon_1)$.

\end{claim}

\begin{proof}
By viewing $|0,b\rangle$ as the initial state, 
if we perform the transformation (\ref{Quantum:SVD1}) to it, then we obtain
$
\sum_{j=1}^n \frac{\beta_j}{\sqrt{2}} ( |w_j^+\rangle |\tilde{\sigma}_j\rangle + |w_j^-\rangle  |-\tilde{\sigma}_j\rangle ).
$
Since $-\tilde{\sigma}_j \leq 0$, we can
add a negative sign to the second term to change it into
$
\sum_{j=1}^n \frac{\beta_j}{\sqrt{2}} ( |w_j^+\rangle |\tilde{\sigma}_j\rangle - |w_j^-\rangle  |-\tilde{\sigma}_j\rangle ).
$
Based on analysis of QPE in Appendix \ref{Quantum phase estimation} about the signs of eigenvalues,
the operation here is feasible.
Now we perform the inverse procedure of the transformation (\ref{Quantum:SVD1}), then
we obtain
$
\sum_{j=1}^n \frac{\beta_j}{\sqrt{2}} ( |w_j^+\rangle - |w_j^-\rangle  )
= \sum_{j=1}^n \beta_j |1\rangle |v_j\rangle.
$
Finally, apply Pauli-$X$ to the first register to generate
$
\sum_{j=1}^n \beta_j |0\rangle |v_j\rangle
=\sum_{j=1}^n \beta_j e^{-2 \pi i \theta_j} |0\rangle |u_j\rangle.
$
The complexity is determined by the implementation of the transformation (\ref{Quantum:SVD1}), which equals $O(T/\epsilon_1)$.
\end{proof}

The unitary $W$ constructed above has eigenvectors $|u_j\rangle$ and eigenvalues $e^{-2\pi i \theta_j}$, so by quantum phase estimation, we can approximate $\theta_j$. As for $\sigma_j$, we can apply quantum singular value estimation (\ref{Quantum:SVD1}) to estimate them.
Now we can state our
quantum algorithm to prepare the state
(\ref{Quantum:EVD}).

\begin{breakablealgorithm}
\label{alg1}
\caption{Quantum eigenvalue estimation of normal matrices}
\begin{algorithmic}[1]
\REQUIRE
(1). An $n\times n$ normal matrix $M$ with (unknown) eigenvalues $\{\sigma_1e^{2\pi i \theta_1},\ldots,\sigma_ne^{2\pi i \theta_n}\}$ and (unknown) eigenvectors $\{|u_1\rangle,\ldots,|u_n\rangle\}$. \\
(2). A quantum state $|b\rangle$ which formally equals $\sum_{j=1}^n \beta_j |u_j\rangle$. 
\\
(3). The precisions $\epsilon_1,\epsilon_2 \in (0,1)$.
\ENSURE The quantum state
\begin{equation}  \label{claimed result}
\sum_{j=1}^n \beta_j |u_j\rangle |\tilde{\sigma}_j\rangle |\tilde{\theta}_j\rangle,
\end{equation} 
where $|\tilde{\sigma}_j- \sigma_j|\leq \epsilon_1$ and $|\tilde{\theta}_j-\theta_j|\leq \epsilon_2$ for all $j$.
\STATE Apply Claim \ref{technical lemma} to construct $W$ defined by equation (\ref{eq5}).
\STATE Apply QPE (see Algorithm \ref{alg:qpe} in Appendix \ref{Quantum phase estimation}) to $W$ with initial state $|0\rangle|b\rangle$.
\STATE Apply quantum SVD (see equation (\ref{Quantum:SVD1})) to estimate the singular values of $M$.
\end{algorithmic}
\end{breakablealgorithm}

In the following, we show more analysis about the steps 2 and 3 of Algorithm \ref{alg1}.

\begin{claim}
The state obtained in step 2 of Algorithm \ref{alg1} is an approximation of
\begin{equation} \label{eq9}
 \sum_{j=1}^n \beta_j |0\rangle |u_j\rangle |\tilde{\theta}_j\rangle.
\end{equation}
The complexity to obtain this state is $O(T/\epsilon_1\epsilon_2)$.
\end{claim}

\begin{proof}
This is a direct application of QPE. More precisely, to apply QPE,
we first generate the state
$
|0\rangle|b\rangle \otimes \frac{1}{\sqrt{m}} \sum_{k=0}^{m-1} |k\rangle =
\frac{1}{\sqrt{m}} \sum_{j=1}^n \sum_{k=0}^{m-1} \beta_j |0\rangle |u_j\rangle  |k\rangle
$
by Hadamard transform,
where $m$ is determined by the precision $\epsilon_2$. Usually, it equals
$O(1/\epsilon_2)$.
Then view $|k\rangle$ as a control qubit to apply $W^k$ to $|0\rangle|b\rangle$, we obtain
$
 \frac{1}{\sqrt{m}} \sum_{j=1}^n \sum_{k=0}^{m-1} \beta_j e^{-2 \pi i \theta_j k} |0\rangle |u_j\rangle |k\rangle.
$
Finally, apply quantum Fourier transform to $|k\rangle$
to generate
$
\frac{1}{m} \sum_{j=1}^n \sum_{l=0}^{m-1} \beta_j
(\sum_{k=0}^{m-1}e^{2\pi i k(\frac{l}{m}-\theta_j)})
 |0\rangle |u_j\rangle  |l\rangle.
$
Similar to the analysis of QPE, the final state  is an approximation of (\ref{eq9}).
The complexity comes from
QPE, which equals $O(T/\epsilon_1\epsilon_2)$.
\end{proof}

In the state (\ref{eq9}), we already obtain approximations of the
phases. As for the singular values,
we can apply the quantum singular value estimation technique, 
which is a simple application of the transformation (\ref{Quantum:SVD1}). 
More precisely,
apply (\ref{Quantum:SVD1})  to
the state (\ref{eq9}), then we have
$
\sum_{j=1}^n \frac{\beta_j}{\sqrt{2}}( |w_j^+\rangle |\tilde{\sigma}_j\rangle + |w_j^-\rangle  |-\tilde{\sigma}_j\rangle)
|\tilde{\theta}_j\rangle.
$
Apply the oracle $|x\rangle|0\rangle \mapsto
|x\rangle||x|\rangle$ to the second register to obtain
$
\sum_{j=1}^n \frac{\beta_j}{\sqrt{2}}( |w_j^+\rangle |\tilde{\sigma}_j\rangle + |w_j^-\rangle  |-\tilde{\sigma}_j\rangle )
|\tilde{\sigma}_j\rangle |\tilde{\theta}_j\rangle.
$
The implementation of this oracle
is discussed in detail in Appendix
\ref{Quantum phase estimation} (see equation (\ref{oracle})).
Finally, apply the inverse of (\ref{Quantum:SVD1}) to yield
the claimed state (\ref{claimed result}).
The complexity is $O(T/\epsilon_1)$. Therefore,
the total cost of Algorithm \ref{alg1} is $O(T/\epsilon_1\epsilon_2)$ as claimed.

\section{Conclusions}

For the problem of estimating eigenvalues,
we proposed a quantum algorithm based on quantum phase estimation and quantum
differential equation solver for diagonalizable matrices 
whose eigenvalues are real,
and
a quantum algorithm based on quantum phase estimation and quantum singular value estimation for normal matrices. 
The output is a state with the first register storing the eigenvalues, the second register storing the corresponding
eigenvectors.
We also generalized the first quantum algorithm to estimate 
all the complex eigenvalues of any diagonalizable matrix.
The complexities of the quantum algorithms are dominated by certain quantum
linear algebraic operations, which are usually exponentially faster than the corresponding classical operations.

This work provides a new attempt to solve the eigenvalue problem in the quantum computer. 
However, there are still many problems that need to be solved.
For instance, the complexities of our quantum algorithms on the condition number $\kappa(E)$ and the precision $\epsilon$ are not optimal. 
There should be a way to improve the dependence into linear. Our quantum algorithm for computing complex eigenvalues is far from optimal. Besides the complexity, it needs the initial state has the form (\ref{initial state}). 
Currently, we do not know how to prepare this kind of initial state.

\section{Acknowledgement}

This work was supported by the QuantERA ERA-NET Cofund in Quantum Technologies implemented within the European
Union's Horizon 2020 Programme (QuantAlgo project), and EPSRC grants EP/L021005/1 and EP/R043957/1.
No new data were created during this study.

\appendix

\section{Quantum phase estimation: brief overview}
\label{Quantum phase estimation}

Quantum phase estimation (QPE) \cite{Kitaev}
is a useful technique to design quantum algorithms. Many quantum algorithms are related to it, 
such as Shor's algorithm \cite{shor1999polynomial} 
and the HHL algorithm to solve linear systems \cite{HHL}.
QPE is an algorithm to estimate the eigenvalues of unitary matrices.
As a simple generalization, 
it can be used to estimate the eigenvalues of
Hermitian matrices \cite{AbramsLloyd} or the singular
values of general matrices \cite{HHL,kerenidis_et_al:LIPIcs:2017:8154}.
In this paper, QPE is also an important technique we will apply to estimate the eigenvalues of more general matrices.
Thus in the following, we briefly review the QPE to estimate the eigenvalues of Hermitian matrices, the error and complexity analyses
can be found in \cite[Section 5.2]{nc}.
One aspect we want to emphasize is
how to determine the signs of the eigenvalues.

\subsection{Apply QPE to estimate the eigenvalues of Hermitian matrices} 

Let $H$ be an $n\times n$ Hermitian matrices with eigenvalues $\lambda_1,\ldots,\lambda_n$ and eigenvectors $|u_1\rangle,\ldots,|u_n\rangle$. Let $|b\rangle$ be any given
quantum state, which can formally be written as a linear combinations
of the eigenvectors $|b\rangle = \sum_{k=1}^n \beta_k |u_k\rangle$.
Suppose we know a upper bound of the eigenvalues, then choose a upper bound $C>0$ such that $|\lambda_j /C|<1/2$.
Let $\epsilon$ be the precision to approximate the eigenvalues
and $\delta$ be the failure probability of the quantum algorithm,
denote $q=
\lceil \log 1/\epsilon \rceil+\lceil \log (2+1/2\delta) \rceil$ and $Q=2^q = O(1/\epsilon\delta)$,
then QPE can be stated as follows:

\vspace{.2cm}

\begin{breakablealgorithm}
\label{alg:qpe}
\caption{Quantum phase estimation (QPE)}
\begin{algorithmic}[1]
\REQUIRE
(1). An $n\times n$ Hermitian matrix $H$ with (unknown) eigenvalues $\{\lambda_1,\ldots,\lambda_n\}$ and (unknown) eigenvectors $\{|u_1\rangle,\ldots,|u_n\rangle\}$. \\
(2). A quantum state $|b\rangle$ which formally equals $\sum_{k=1}^n \beta_k |u_k\rangle$. 
\\
(3). The precision $\epsilon$ and failure probability $\delta$.
\ENSURE The quantum state
\be
\sum_{k=1}^n \beta_k |\tilde{\lambda}_k\rangle |u_k\rangle ,
\ee
where $|\tilde{\lambda}_k- \lambda_k|\leq \epsilon$ for all $k$.
\STATE Set the initial state  as
\be
|\psi_0\rangle = \frac{1}{\sqrt{Q}} \sum_{j=0}^{Q-1}|j\rangle |b\rangle.
\ee

\STATE Apply control operator $\sum_{j=0}^{Q-1} |j\rangle \langle j| \otimes 
e^{2 \pi i j H /C}$ to $|\psi_0\rangle$ to prepare
\be
\label{QPE:required state}
|\psi_1\rangle = \frac{1}{\sqrt{Q}} \sum_{j=0}^{Q-1} \sum_{k=1}^n \beta_k e^{2 \pi i j \lambda_k /C } |j\rangle |u_k\rangle.
\ee

\STATE Apply quantum inverse Fourier transform to the first register of $|\psi_1\rangle$
\be
|\psi_2\rangle = \frac{1}{Q}\sum_{l=0}^{Q-1}\sum_{k=1}^n \beta_k  \sum_{j=0}^{Q-1} e^{2 \pi i j (\frac{\lambda_k}{C} - \frac{l}{Q})} |l\rangle |u_k\rangle.
\ee
\end{algorithmic}
\end{breakablealgorithm}

\vspace{.2cm}

Denote
\bea
\Lambda_k &=& 
 \begin{cases} 
    \displaystyle \{l\in \{0,1,\ldots,Q-1\}: |\frac{\lambda_k}{C}-\frac{l}{Q}|\leq \epsilon \},   & \text{if $\lambda_k\geq 0$,} \\
   \displaystyle  \{l\in \{0,1,\ldots,Q-1\}: |1+\frac{\lambda_k}{C}-\frac{l}{Q}|\leq \epsilon \}, & \text{if $\lambda_k<0$.}
  \end{cases}
  \\
 |\Lambda_k\rangle &=& \frac{1}{Q}  \sum_{l \in \Lambda_k}
\sum_{j=0}^{Q-1} e^{2 \pi i j (\frac{\lambda_k}{C} - \frac{l}{Q})}
\, |l\rangle.
\eea
In theory, we do not know $\Lambda_k$. The notation introduced here is to simplify the expression of $|\psi_2\rangle$. 
Then we can rewrite $|\psi_2\rangle$ as
\be
|\psi_2\rangle = \sum_{k=1}^n \beta_k |\Lambda_k\rangle
|u_k\rangle + {\rm others}.
\ee
It is shown that the amplitude of the 
first term is larger than $\sqrt{1-\delta}$. 
Thus, if we choose $\delta$ small enough (e.g. 0.01 or $1/{\rm poly}\log n$), then we can approximately  write $|\psi_2\rangle$ as
\be \label{eq-qpe}
|\psi_2\rangle \approx \sum_{k=1}^n \beta_k |\Lambda_k\rangle
|u_k\rangle.
\ee
Any integer $l$ in $\Lambda_k$ provides an $\epsilon$-approximation 
$l/Q$ of $\lambda_k/C$ or $1+\lambda_k/C$. 

One problem we need pay attention to the above procedure is the signs of the eigenvalues. In $|\psi_1\rangle$, if $\lambda_k\geq 0$, then
the coefficient of $|j\rangle|u_k\rangle$ is $\beta_k e^{2 \pi i j \lambda_k/C }$. However, if $\lambda_k < 0$, then
the coefficient becomes $\beta_k e^{2 \pi i j (1+\lambda_k /C) }$ as we need to make sure the phase lies between 0 and $2\pi$.
As a result, in QPE we actually obtain approximations of $\lambda_k /C$ if $\lambda_k \geq 0$ and of $1+\lambda_k /C$ if $\lambda_k < 0$.
Note that $C$ is chosen such that $|\lambda_k /C|<1/2$, thus in the former case, if $l/Q$ is an approximation of $\lambda_k /C$, then $l\leq Q/2$.
In the latter case, if $l/Q$ is an approximation of $1+\lambda_k /C$, then
$(Q-l)/Q$ is an approximation of $-\lambda_k /C$. Hence $l>Q/2$.
In conclusion, if integers in $\Lambda_k$ are smaller 
than or equal to $Q/2$,
then we know $\lambda_k\geq 0$, otherwise  $\lambda_k<0$.
Therefore, based on the integers in $\Lambda_k$, we can determine the signs of the eigenvalues.

The complexity of QPE is dominated by the second step.
If the complexity to implement the unitary $e^{i H t'}$ to precision $\epsilon'$ in the quantum circuit is $T(t',\epsilon',n)$, then
the complexity of QPE to compute $\epsilon$-approximations of eigenvalues of Hermitian matrices is 
$T(C/\epsilon \delta,\epsilon',n)$.
For example, for an $s$ sparse $n\times n$ Hermitian matrix 
$H=(h_{ij})$, it is shown in
\cite{low2017optimal} that
$T(t',\epsilon',n)= O((st'\|H\|_{\max} + \frac{\log 1/\epsilon'}{\log\log 1/\epsilon'}) \log n)$, which is optimal
at the parameters $t'$ and $\epsilon'$. Here $\|H\|_{\max}=\max_{i,j}|h_{ij}|$.
In this case, the complexity of QPE is
$O(sC\|H\|_{\max}(\log n)/\epsilon\delta)$ if we set $\epsilon'=\epsilon$.

In the following of this paper, 
we will simply write equation (\ref{eq-qpe}) as 
\be \label{eq:QPE state}
\sum_{k=1}^n \beta_k |\tilde{\lambda}_k\rangle |u_k\rangle.
\ee
Although this state is not rigorous especially when $\lambda_k<0$,
it clearly describes the result of QPE.
It can be viewed as the quantum eigenvalue decomposition of Hermitian matrices.
Moreover, for simplicity we will ignore $\delta$ in the complexity
analysis by setting it as a small constant.

At the end of this part, we give a method to
implement the following transformation
\be \label{oracle}
\sum_{k=1}^n \beta_k |\tilde{\lambda}_k\rangle |u_k\rangle |0\rangle
\mapsto
\sum_{k=1}^n \beta_k |\tilde{\lambda}_k\rangle |u_k\rangle ||\tilde{\lambda}_k|\rangle,
\ee
which will be useful in Section \ref{Normal Matrices}.
Define the function $f:\mathbb{Z}_Q\rightarrow\mathbb{Z}_Q$ by
\be
f(l) =   \begin{cases}
    l   & \text{if $l\leq Q/2$,} \\
    Q-l & \text{if $l> Q/2$.}
  \end{cases}
\ee
It defines an oracle
$
U_f:|x,y\rangle \mapsto |x,y\oplus f(x)\rangle.
$
Based on this oracle, we can implement
(\ref{oracle}) via
\be
\sum_{k=1}^n \beta_k \sum_{l \in \Lambda_k}\Lambda_{kl} \, |l\rangle |u_k\rangle|0\rangle
\mapsto
\sum_{k=1}^n \beta_k \sum_{l \in \Lambda_k} \Lambda_{kl} \, |l\rangle
|u_k\rangle |f(l)\rangle.
\ee
The equivalence comes naturally from the above analysis about the signs of eigenvalues. If $\lambda_k\geq 0$, then 
$|\lambda_k|=\lambda_k$ and $l\leq Q/2$ for all $l\in \Lambda_k$.
Thus $f(l)=l$.
If $\lambda_k< 0$, then $|\lambda_k|=-\lambda_k$
and $l>Q/2$ for all $l\in \Lambda_k$.
Moreover, $(Q-l)/Q$ are approximations of $\lambda_k/C$
for all $l\in \Lambda_k$, thus $f(l)=Q-l$.

\subsection{Apply QPE to estimate the singular values of matrices} 

For any matrix $M=(m_{ij})_{m\times n}$, in a quantum computer we can 
compute its singular value decomposition (SVD). More precisely,
assume that its SVD is $M=\sum_{j=1}^d \sigma_j |u_j\rangle \langle v_j|$, where $d=\min\{m,n\}$. Denote
\be
\widetilde{M} = \left(\begin{array}{cc}
0 & M  \\
M^\dag & 0
\end{array}\right),
\ee
which is Hermitian.
The eigenvalues of $\widetilde{M}$ are $\{\pm \sigma_j: j=1,\ldots,d\}$.
The corresponding eigenvectors are
\begin{equation}
 |w_j^\pm\rangle
:=\frac{1}{\sqrt{2}} (|0\rangle|u_j\rangle \pm |1\rangle|v_j\rangle),
\quad
j=1,\ldots,d.
\end{equation}
Based on QPE, we can implement the following transformation
(see equation (\ref{eq:QPE state}))
\begin{equation} \label{Quantum:SVD}
\sum_{j=1}^d \beta_j^+ |w_j^+\rangle + \beta_j^- |w_j^-\rangle
\mapsto \sum_{j=1}^d \beta_j^+ |w_j^+\rangle |\tilde{\sigma}_j\rangle + \beta_j^- |w_j^-\rangle |-\tilde{\sigma}_j\rangle
\end{equation}
in a quantum computer  \cite{block-encoding}, where $|\sigma_j-\tilde{\sigma}_j|\leq \epsilon$. 
Based on the analysis about QPE, the minus sign in equation (\ref{Quantum:SVD}) is reasonable.

If $M$ is $s$ sparse, then
we can implement (\ref{Quantum:SVD}) in cost
$O(sC\|M\|_{\max}(\log
(m+n))/\epsilon)$.
Finally, we conclude the above analysis into the following proposition.

\begin{prop}
\label{prop:qsvd}
Let $M$ be an $m\times n$ matrix.
Let $C$ be a upper bound of its singular values.
If $M$ is $s$ sparse, then
we can implement (\ref{Quantum:SVD}) in time
$O(\epsilon^{-1}sC\|M\|_{\max}\log
(m+n))$.
\end{prop}

\section{Verification of equations (\ref{verify1})-(\ref{verify3})}
\label{verification}
Recall from (\ref{coeff matrix}), (\ref{eq:approsolu}) that
\begin{eqnarray*}
C_{m,k}(2\pi iM\Delta t) &=& \sum_{p=0}^{m(k+1)} |p\rangle \langle p|\otimes I - \sum_{p=0}^{m-1} \sum_{q=1}^{k} |p(k+1)+q\rangle\langle p(k+1)+q-1|\otimes \frac{2\pi iM\Delta t}{q} \\
&& -\, \sum_{p=0}^{m-1} \sum_{q=0}^k |(p+1)(k+1)\rangle\langle p(k+1)+q|\otimes I, \\
\tilde{\x} &=& \sum_{p=0}^{m-1} \sum_{q=0}^k |p(k+1)+q\rangle |\x_{p,q}\rangle + |m(k+1)\rangle |\x_{m,0}\rangle.
\end{eqnarray*}
Then
\[
C_{m,k}(2 \pi iM\Delta t) |m(k+1)\rangle |\x_{m,0}\rangle
=|m(k+1)\rangle |\x_{m,0}\rangle.
\]
Moreover,
\begin{eqnarray*}
&& C_{m,k}(2\pi iM\Delta t) \sum_{p=0}^{m-1} \sum_{q=0}^k |p(k+1)+q\rangle |\x_{p,q}\rangle \\
&=& \sum_{p=0}^{m-1} \sum_{q=0}^k |p(k+1)+q\rangle |\x_{p,q}\rangle 
- \sum_{p=0}^{m-1} \sum_{q=1}^{k} |p(k+1)+q\rangle \otimes \frac{2\pi iM\Delta t}{q} |\x_{p,q-1}\rangle  \\
&& - \, \sum_{p=0}^{m-1}  |(p+1)(k+1)\rangle \sum_{q=0}^k |\x_{p,q}\rangle \\
&=& |0\rangle |\x_{0,0}\rangle - |m(k+1)\rangle \sum_{q=0}^k |\x_{p,q}\rangle + \sum_{p=1}^{m-1} |p(k+1)\rangle \left(|\x_{p,0}\rangle - \sum_{q=0}^k |\x_{p-1,q}\rangle\right)\\
&& + \, \sum_{p=0}^{m-1} \sum_{q=1}^{k} |p(k+1)+q\rangle \left( |\x_{p,q}\rangle - \frac{2\pi iM\Delta t}{q} |\x_{p,q-1}\rangle \right).
\end{eqnarray*}
Consequently, $|\x_{0,0}\rangle = |E_j\rangle$.
If $0\leq p\leq m-1$ and $1\leq q \leq k$, then
\[
|\x_{p,q}\rangle = \frac{2\pi iM\Delta t}{q} |\x_{p,q-1}\rangle = \frac{(2\pi iM\Delta t)^q}{q!} |\x_{p,0}\rangle.
\]
If $0\leq p \leq m$ and $q=0$, then
\[
|\x_{p,0}\rangle = \sum_{q=0}^k |\x_{p-1,q}\rangle
=\sum_{q=0}^k \frac{(2\pi iM\Delta t)^q}{q!} |\x_{p-1,0}\rangle = T_k(2\pi iM\Delta t) |\x_{p-1,0}\rangle.
\]

\section{Proof of Lemma \ref{distribution}}\label{uniform-distribution}
First we consider the special case of uniform distribution.
Denote $|\phi\rangle = \frac{1}{\sqrt{n}} \sum_{j=1}^n |j\rangle$. 
This corresponds to a uniform distribution, and $|j\rangle$ can be viewed as the $j$-th event.
Consider $|\phi\rangle^{\otimes m}$. 
Perform measurements on the basis states, then the probability to obtain all the basis states $|j_1,\ldots,j_m\rangle$ such that $\{j_1,\ldots,j_m\}=\{1,\ldots,n\}$ is
\[
P:=\frac{1}{n^m} \sum_{\substack{i_1+\cdots+i_n=m \\i_1,\cdots,i_n\geq 1}}\binom{m}{i_1,\ldots,i_n},
\]
where
\[
\binom{m}{i_1,\ldots,i_n}:=\frac{m!}{i_1!\cdots i_n!}
\]
is multinomial coefficient.

Let $x_1,\ldots,x_n$ be $n$ variables, then
\[
(x_1+\cdots+x_n)^m=\sum_{ i_1+\cdots+i_n=m}\binom{m}{i_1,\ldots,i_n} x_1^{i_1}\cdots x_n^{i_n}.
\]
Especially when $x_1=\cdots=x_n=1$, we have
\[
n^m=\sum_{ i_1+\cdots+i_n=m}\binom{m}{i_1,\ldots,i_n}.
\]
Note that $Pn^m$ refers to the number of monomials in $(x_1+\cdots+x_n)^m$ such that all $x_i$ appear. Thus,
\[
\sum_{\substack{i_1+\cdots+i_n=m \\i_1,\cdots,i_n\geq 1}}\binom{m}{i_1,\ldots,i_n}
\geq n^m-n(n-1)^m.
\]
The right hand side of the above inequality means that we set
$x_i=0$ and the remaining variables as 1 in $(x_1+\cdots+x_n)^m$. There are $n$ possibilities. However, it may happen that
different cases have common terms. Thus
\[
P\geq \frac{n^m-n(n-1)^m}{n^m}
=1-n(1-\frac{1}{n})^m
\approx 1 - n e^{-m/n}.
\]
Choose $m$ such that $n e^{-m/n} = \delta$ is small (say $0.01$), that is $m=n\log(n/\delta)$, then
$P\geq 1-\delta$.
The above analyses show that in a uniform distribution, to make sure every event happen, it suffices to make $m=O(n\log n)$ experiments.

In the general case, assume that $|\phi\rangle=\sum_{j=1}^n \sqrt{p_j}|j\rangle$, 
where $p_1+\cdots+p_n=1$. Also consider $|\phi\rangle^{\otimes m}$, then
\[
P=\sum_{\substack{i_1+\cdots+i_n=m \\i_1,\cdots,i_n\geq 1}}\binom{m}{i_1,\ldots,i_n} p_1^{i_1}\cdots p_n^{i_n}.
\]
Denote $p_{\max}=\max\{p_1,\ldots,p_n\}$. We assume that $p_{\max}\neq 1$, otherwise the distribution is trivial.
Then similar analysis shows that
\[
P\geq 1 - \sum_{j=1}^n (1-p_j)^m
\geq 1-n(1-p_{\max})^m
\approx 1- ne^{-mp_{\max}}.
\]
Thus it suffices to choose $m=p_{\max}^{-1}\log(n/\delta)$.
Since $\sum_j p_j=1$, we have $p_{\max}\geq 1/n$. Consequently, $m\leq n\log(n/\delta)$.


\bibliographystyle{plain}
\bibliography{main}

\end{document}